\documentclass[aap]{imsart}
\usepackage[utf8]{inputenc}
\usepackage[english]{babel}
\usepackage{amsmath,amssymb,verbatim,mathtools,amsthm,bbm,mathrsfs}
\usepackage[numbers,sort&compress]{natbib}
\usepackage{indentfirst}
\usepackage{natbib}
\usepackage{mathrsfs}
\RequirePackage[colorlinks,citecolor=blue,urlcolor=blue]{hyperref}
\usepackage{enumitem}
\usepackage[normalem]{ulem}
\usepackage{xfrac}
\newlength\tindent
\usepackage{tikz} 
\usetikzlibrary{automata}
\usetikzlibrary{positioning}
\usetikzlibrary{calc,arrows.meta,positioning}
\usepackage[ruled,vlined]{algorithm2e}
\usepackage{float}
\usepackage{algorithmic}
\usepackage[md]{titlesec}

\newtheorem{theorem}{Theorem}
\newtheorem{proposition}{Proposition}
\newtheorem{lemma}{Lemma}
\newtheorem{definition}{Definition}
\newtheorem*{notation*}{Notation}
\newtheorem{corollary}{Corollary}
\newtheorem{remark}{Remark}

\newtheorem{assump}{Assumption}

\DeclareMathOperator*{\argmin}{arg\,min}
\DeclareMathOperator{\E}{\mathbb{E}}

\DeclareMathOperator{\Prob}{\mathbb{P}}

\DeclareMathOperator{\ind}{\perp \!\!\! \perp }
\newcommand{\norm}[1]{\left\lVert#1\right\rVert}

\newcommand{\1}{\mathbbm{1}}

\newcommand{\abs}[1]{\left|#1\right|}
\newcommand{\iidsim}{\overset{\mbox{\tiny{iid}}}{\sim}}

\newcommand{\A}{\mathcal{A}}
\newcommand{\Gv}{\mathcal{G}_v}
\newcommand{\Ev}{\mathcal{E}_v}
\newcommand{\GvInd}{\1_{\Gv}(X_{1:v-1}^{1:N},\bar{X}_{1:v-1}^{1:N})}

\newcommand{\scs}[1]{\textcolor{red}{#1}}
\newcommand{\scsfootnote}[1]{\textcolor{red}{\footnote{\textcolor{red}{#1}}}}
\newcommand{\scsfootnotemark}{\textcolor{red}{\footnotemark}}
\newcommand{\scsfootnotetext}[1]{\textcolor{red}{\footnotetext{\textcolor{red}{#1}}}}
\newcommand{\scsout}[1]{\scs{\sout{\textcolor{black}{#1}}}}

\newcommand{\jrm}[1]{\textcolor{blue}{#1}}

\begin{document}

\begin{frontmatter}
\title{Finite Sample Complexity of Sequential Monte Carlo Estimators on Multimodal Target Distributions}
\runtitle{Finite Sample Complexity of SMC on Multimodal Distributions}

\begin{aug}
\author[A]{\fnms{Joseph}~\snm{Mathews}\ead[label=e1]{joseph.mathews@duke.edu}}
\and
\author[B]{\fnms{Scott C.}~\snm{Schmidler}\ead[label=e2]{scott.schmidler@duke.edu}\orcid{0000-0000-0000-0000}}
\address[A]{Department of Statistical Science, Duke University \printead[presep={,\ }]{e1}}

\address[B]{Department of Statistical Science,
Duke University \printead[presep={,\ }]{e2}}
\end{aug}

\begin{abstract}
    We prove finite sample complexities for sequential Monte Carlo (SMC) algorithms which require only \textit{local} mixing times of the associated Markov kernels.   Our bounds are particularly useful when the target distribution is multimodal and global mixing of the Markov kernel is slow; in such cases our approach 
    establishes benefits of 
    SMC over the corresponding Markov chain Monte Carlo (MCMC) estimator.
    The 
    lack of global mixing is addressed by sequentially controlling the bias introduced by SMC resampling procedures. 
    We apply these results to obtain complexity bounds for approximating expectations under mixtures of log-concave distributions, and show that SMC provides a fully polynomial time randomized approximation scheme for some difficult multimodal problems where the corresponding Markov chain sampler is exponentially slow. Finally, we compare the bounds obtained by our approach to existing bounds for tempered Markov chains on the same problems.
\end{abstract}



\end{frontmatter}

\section{Introduction}
Approximating integrals with respect to a complicated, high-dimensional probability distribution $\pi$ is an important problem spanning multiple disciplines, such as  Bayesian statistical inference, machine learning,  statistical physics, and theoretical computer science \cite{BDA, Levin_Peres}. Sequential Monte Carlo (SMC) methods are a large class of stochastic approximation algorithms designed to solve these problems by combining Markov chain Monte Carlo (MCMC) methods and resampling strategies to sequentially sample from a series of probability distributions. Some examples of SMC algorithms include population Monte Carlo methods \cite{Cappe}, annealed importance sampling \cite{annealed}, sequential particle filters \cite{chopin_static}, and population annealing \cite{PA_2}, among many others \cite{jasra_PMC}. Closely related - but purely MCMC - methods include parallel tempering (PT) \cite{geyer} and simulated tempering (ST) \cite{Marinari}, which have been referred to as population-based MCMC methods \cite{jasra_PMC}.

An SMC sampler is generally constructed as follows. The target distribution $\pi$ is embedded in a sequence of distributions $\mu_{0},\ldots,\mu_{V-1},\mu_{V} = \pi$ chosen such that $\mu_{0}$ is easy to sample from and $\mu_{v-1}$ is `close' to $\mu_{v}$ for $v = 1,\ldots,V$. Random variables (called \textit{particles}) are initially sampled from $\mu_{0}$ and are propagated through the sequence of distributions  sequentially via a combination of time-inhomogeneous MCMC and resampling moves until step $V$, where the resulting particles can be used to approximate expectations under $\pi$. For a comparison of various resampling strategies, see \cite{douc}. For a discussion of different Markov kernels used for the MCMC (\textit{mutation}) step, see \cite{delmoral_state}. For adaptive strategies, see  \cite{beskos} and \cite{fearnhead}. An explicit statement of the algorithm considered in this paper is given in Section \ref{smc}.

Two properties of SMC samplers make them particularly attractive to practitioners. First, they can be parallelized easily, improving their scalability to high-dimensional problems \cite{geweke}. Second, SMC samplers are believed to perform well when $\pi$ has multiple, well-separated modes. In these situations, standard MCMC methods such as random walk Metropolis \cite{metropolis}, Hamiltonian Monte Carlo \cite{neal_HMC} and Gibbs samplers, often exhibit metastable behavior, leading to slow mixing and poor approximation performance (see e.g. \cite{mangoubi}). SMC samplers are frequently applied to such multimodal targets across different disciplines, see e.g. \cite{geweke, nested_smc, wan_and_zabaras} for applications to Bayesian statistics and \cite{jasra, PA_parallel, PA_2} for applications to spin models in statistical physics. The performance of SMC on multimodal problems
has also been tested via simulation studies on various model problems in \cite{jasra_PMC, rudoy, tan}. See also \cite{PA_PT_SMC} for simulation studies comparing the performance of SMC methods to PT and ST algorithms on the Potts model. 

However, theoretical results in support of these purported benefits of SMC over MCMC are limited.  For one thing, such comparisons require quantitative analysis of the number of particles or steps required by each algorithm on a given problem.   
But while the asymptotic properties of SMC samplers have been explored extensively in the literature (for example, \cite{chopin_clt} establishes a central limit theorem under various resampling strategies, \cite{adaptive_clt} proves a central limit theorem for SMC methods using adaptive schemes, \cite{douc_lln} proves a law of large numbers, and \cite{jasra} bound the asymptotic variance of SMC for multimodal problems), rigorous results on finite-sample performance are much more limited. Results for finite particle sets have been obtained by use of Feynman-Kac formulations \cite{whiteley, schweizer, schweizer_mm, eberle_1, eberle_2}; however, these bounds typically involve expectations of the corresponding Feynman-Kac propagators, which are generally not available explicitly, and assume mixing conditions 
on the associated MCMC kernels which are difficult to verify on general state spaces.  These factors make it difficult to determine the dependence of SMC algorithms on key quantities of interest such as the problem dimension. More recently, \cite{marion} established finite sample bounds 
using a coupling construction. However, their results depend on the global mixing rates of the mutation kernels, and problems where these are controllable (e.g when polynomial rather than exponential in problem size) are likely to be precisely those cases where MCMC performs well without SMC. 

In this paper we establish rigorous conditions guaranteeing the performance of SMC samplers on multimodal target distributions by proving upper bounds on the computational complexity.  These results are obtained by establishing upper bounds on the number of particles and MCMC steps needed to closely approximate expectations under the target distribution with high probability. Our approach builds on that of \cite{marion}, but here we emphasize settings where global mixing may be slow, as in multimodal situations, but where the kernel can be assumed to mix \textit{locally}.  In particular, our approach assumes particles are only mutated through restricted Markov kernels, corresponding to an arbitrary partition of the state space. This `local mixing' approach is similar to \cite{schweizer_mm, jasra},  but here we use this approach to obtain finite sample bounds rather than asymptotic ones.
We show that when SMC accurately approximates `local' expectations (i.e. expectations restricted to each partition element), global expectations can also be approximated accurately. This approach also calls to mind
the previous mixing time analyses of the PT and ST algorithms in \cite{woodard_1} (see also \cite{holden}), who used the state space decomposition approach of \cite{madras}. 
We apply our bounds to mixtures of log-concave distributions and the mean field Ising model. In addition, we use our bounds to provide a rigorous comparison of the SMC, PT, and ST algorithms.

While our approach builds on the coupling argument of \cite{marion}, the inability to rely on global mixing substantially complicates the analysis. Unlike in \cite{marion}, particles can no longer be coupled to a set of exact \textit{i.i.d.} random variates, as the bias introduced by resampling is no longer `corrected' automatically by a globally mixing kernel, and instead must be controlled directly. Consequently our analysis highlights the benefits obtained from the resampling step, which indeed is crucial to the success of the algorithm when the mutation kernel cannot be assumed to mix globally at each step.

This paper is organized as follows. Section \ref{problem_setup} establishes notation, formalizes the SMC algorithm we are considering, and outlines our proof approach. Section \ref{results} contains our main theorem and supporting results. Section \ref{applications}, applies these results to obtain bounds for mixtures of log-concave distributions and the mean field Ising model. Section \ref{smc_pt_compare} compares our bounds to those obtained by   \cite{woodard_1,holden}
for the parallel and simulated tempering algorithms. Section \ref{conclusion}, concludes with a discussion of future directions.

\section{Problem Setup}\label{problem_setup}
Let $\pi$ be a target distribution defined on a state space $\mathcal{X}$ with corresponding $\sigma$-field $\mathcal{B}$ and dominating measure $\rho$. In a slight abuse of notation we use the same letter to denote both the distribution and its corresponding density (e.g. the density of $\pi$ is $\pi(x)$). Let $\mathcal{A} = \{A_{j}: j=1,\ldots,p \}$ be a partition of $\mathcal{X}$. Typically $\mathcal{A}$ will represent a partition of $\mathcal{X}$ into local modes of $\pi$, although our results apply for any partition.  While $\mathcal{A}$ is usually not specified in practice, the algorithm given in Section~\ref{smc} serves as a useful theoretical model for understanding the behavior of SMC algorithms when $\pi$ is multimodal.

\subsection{Sequential Monte Carlo (SMC)}
\label{smc}
Let $\mu_{0}, \ldots, \mu_{V}$ be a sequence of distributions defined on $\mathcal{X}$ with dominating measure $\rho$ and $\mu_{V} = \pi$. We write $\mu_{v}(x) = q_{v}(x) / z_{v}$ with $q_{v}(x)$ denoting the unnormalized density and $z_{v}$ the (often unknown) normalizing constant. Define the \textit{importance weight} $w_{v}(x) := q_{v}(x)/q_{v-1}(x)$. We assume that the density ratios are uniformly bounded.
\begin{assump}\label{assumption: density ratios}
\begin{align*}
    \frac{\mu_{v}(x)}{\mu_{v-1}(x)} = \frac{z_{v-1}}{z_{v}} \cdot \frac{q_{v}(x)}{q_{v-1}(x)} = \frac{z_{v-1}}{z_{v}} \cdot w_{v}(x) \leq ZW, 
\end{align*}
 where $Z = \max_{v \in \{0,\ldots,V\}} \frac{z_{v-1}}{z_{v}}$ and $W =\max_{v \in \{0,\ldots,V\}} \sup_{x \in \mathcal{X}}  w_{v}(x)$.
\end{assump}
Denote the minimum probability of all partition elements across all steps as $\mu^{\star} := \min_{v = 0,\ldots,V} \min_{j=1,\ldots,p} \mu_{v}(A_{j})$. The algorithm we analyze is defined by the following procedure:
\begin{enumerate}
    \item For $v = 0$, draw $X^{1:N}_{0} \overset{iid}{\sim} \mu_{0}$
    \item For $v = 1,\ldots,V$,
    \begin{enumerate}
        \item Set $\tilde{X}^{i}_{v} = X^{k}_{v-1}$ with probability $w_{v}(X^{k}_{v-1}) /   \sum^{N}_{m=1} w_{v}(X^{m}_{v-1})$ independently for  $i=1,\ldots,N$.
        \item Draw $X^{i}_{v} \sim K_{v|A_{\xi^{i}_{v}}}^{t}(\tilde{X}^{i}_{v}, \cdot)$ for $i = 1,\ldots,N$ and $t \geq 1$,
    \end{enumerate}
\end{enumerate}
where $\xi^{i}_{v} = \sum^{p}_{j=1} j \mathbbm{1}_{A_j}(\tilde{X}_{v}^{i})$ is the (random) partition element index containing $\tilde{X}_{v}^i$, and 
\begin{align}\label{Eqn: restricted kernel}
    K_{v|A}(x,B) = K_{v}(x,B) + K_{v}(x,A^{c})\mathbbm{1}_{x \in B} \text{ for }x \in A, B \subset A,
\end{align}
where $K_{v}$ is a $\mu_{v}$-invariant Markov kernel and $K_{v\mid A}$ is the restriction of $K$ to $A\subset\mathcal{X}$. Note that $K_{v\mid A}$ is $\mu_{v\mid A}$-invariant by construction and is assumed to be ergodic for any $A\subset \mathcal{X}$, where
\begin{align*}
    \mu_{v|A}(B) = \frac{\mu_{v}(A \cap B)}{\mu_{v}(A)}, \text{ for } B \subset \mathcal{X}
\end{align*}
The use of $K_{v \mid A_{j}}$ instead of $K_{v}$ implies that the mutation kernel cannot cross between modes/partition elements. Consequently, we have that $X^{i}_{v} \in A_{j}$ if and only if $\tilde{X}^{i}_{v} \in A_{j}$.
Each iteration of the algorithm above generates a set of \textit{resampled} particles $\tilde{X}^{1:N}_{v} := (\tilde{X}^{1}_{v}, \ldots, \tilde{X}^{N}_{v})$ and \textit{mutated} particles $X^{1:N}_{v} := (X^{1}_{v}, \ldots, X^{N}_{v})$. The full algorithm then produces a total of $N(2V+1)$ particles defined on $\mathcal{X}^{2N(V+1)}$ that evolve according to the following joint probability laws,
\begin{align}
    \Prob_{0}(X^{1:N}_{0} \in B_{1} \times \ldots \times B_{N}) &= \prod^{N}_{i=1} \mu_{0}(B_{i}) \label{joint_intial}\\
    \tilde{\Prob}_{v}(\tilde{X}^{1:N}_{v} \in B_{1} \times \ldots \times B_{N} \mid X^{1:N}_{v-1} ) &= \prod^{N}_{i=1} \sum^{N}_{k=1} \frac{w_{v}(X^{k}_{v-1}) \mathbbm{1}_{B_{i}}(X^{k}_{v-1})}{N \sum^{p}_{r=1}\hat{w}^{r}_{v}} \label{joint_resample} \\
    \Prob_{v}(X^{1:N}_{v} \in B_{1} \times \ldots \times B_{N} \mid \tilde{X}^{1:N}_{v}) &= \prod^{N}_{i=1} K_{v | A_{\xi^{i}_{v}}}^{t}(\tilde{X}^{i}_{v}, B_{i}), \label{joint_mutate}
\end{align}
where $\hat{w}^{j}_{v} = N^{-1} \sum^{N}_{i=1} w_{v}(X^{i}_{v-1})\mathbbm{1}_{A_j}(X_{v-1}^j)$ and $B_{i} \subset \mathcal{X}$ for $i=1,\ldots,N$.  

\subsection{Mixing times} A distribution $\nu$ is said to be \textit{$M$-warm} with respect to a distribution $\mu$ if $ \sup_{B \subset \mathcal{X}}\nu(B) / \mu(B) \leq M$. Let $\mathcal{P}_{M}(\mu)$ denote the set of all such distributions.  Define the $M$-warm mixing time for $j = 1,\ldots,p$ by: 
\begin{align} \label{mixing_time}
    \tau_{v,j}(\epsilon, M) := \min \left\{t : \sup_{\eta \in \mathcal{P}_{M}(\mu_{v|A_{j}})} \norm{\int_{A_{j}}\eta_{|A_{j}}(dx)K_{v|A_{j}}^{t}(x, \cdot) - \mu_{v|A_{j}}(\cdot)}_{\text{TV}} \leq \epsilon   \right\}
\end{align}
Assuming that the initial distribution is $M$ warm is common among MCMC mixing time results (see e.g. \cite{yuansi,lovasz_vempala_LC,keru}); however, choosing such an initial distribution can be difficult in practice. 

\subsubsection{Locally Warm Distributions}
In later sections we will require a stronger notion of warmness than the standard definition above.  We call this stronger version {\it local warmness}: 
\begin{definition}\label{def:locally warm}
A distribution $\nu$ is \emph{locally $M$-warm} with respect to a distribution $\mu$ and partition $\mathcal{A} = \{A_{j}: j=1,\ldots,p\}$ if,
\begin{align*}
    \sup_{B \subset A_{j}} \frac{\nu_{|A_{j}}(B)}{\mu_{|A_{j}}(B)} \leq M 
\end{align*}
\end{definition}
\noindent Note that this is a stronger notion of warmness since it requires that $\tilde{\mu}_{v}$ assign sufficiently large probability to each partition element. Local warmness is useful when considering potentially multimodal target distributions, because warmness of an initial distribution by the standard definition is insufficient to guarantee rapid mixing on multimodal distributions. For example, consider the Gaussian mixture $\mu(x) = \frac{1}{2}\mathcal{N}(x ;-\mathbf{1}_{d},I_{d}) + \frac{1}{2}\mathcal{N}(x ;\mathbf{1}_{d},\sigma^2I_{d})$ studied by \cite{woodard_2}. The initial distribution $\nu(x) = \mathcal{N}(x ; -\mathbf{1}_{d}, I_{d})$ provides a $2$-warm start but, as the authors show, a random-walk Metropolis-Hastings algorithm still produces a torpidly mixing chain.
A key step in our analysis will be to show that the resampling step (2a) of the SMC algorithm produces a distribution $\tilde{\mu}$ that is locally warm for arbitrary partitions.
\subsection{Coupling Constructions}
\label{coupling_overview}
The proof of Lemma~\ref{coupling} involves the construction of coupled random variables. We briefly introduce some notation for this here but defer the detailed construction to Appendix~\ref{appendix_a}. Our analysis constructs coupled pairs of random variables $(X^{i}_{v}, \bar{X}^{i}_{v})$. The pairs are conditionally independent given $\tilde{X}^{i}_{v}$ and satisfy for $i = 1,\ldots,N$ and $j = 1,\ldots,p$,
\begin{enumerate}[label=(\alph*)]
    \item $\Prob(X^{i}_{v} \in B \mid \tilde{X}^{i}_{v} \in A_{j}) = \hat{\mu}_{v|A_{j}}(B)$
    \item $\Prob(\bar{X}^{i}_{v} \in B \mid \tilde{X}^{i}_{v} \in A_{j}) = \mu_{v|A_{j}}(B)$
    \item $\Prob(X^{i}_{v} \neq \bar{X}^{i}_{v} \mid \tilde{X}^{i}_{v} \in A_{j}) = ||\hat{\mu}_{v | A_{j}}(\cdot) - \mu_{v | A_{j}}(\cdot)||_{\text{TV}}$
\end{enumerate}
While the $X_v^i$'s represent the  particles produced by the SMC algorithm as described above, the coupled random variables $\bar{X}^{1:N}_{v} := (\bar{X}^{1}_{v}, \ldots, \bar{X}^{N}_{v})$ are simply a theoretical tool for our analysis and are not constructed in practice. Note that the construction preserves $\mathcal{L}(X^{i}_{v})$ and so the behavior of the algorithm given in Section \ref{smc} is unchanged (see Section \ref{appendix_a}). 

\subsection{Probability space and Marginal Distributions}
Going forward, it should be understood that expectations are taken with respect to the joint probability measure defined by (\ref{joint_intial})-(\ref{joint_mutate}) and the coupling constructions up to time $V$.  We simply write a generic $\Prob$ (and similarly for $\E$) to denote this joint probability measure. The corresponding state space is $\mathcal{X}^{N(3V + 1)}$ with product $\sigma$-field $\mathcal{B}^{N(3V + 1)}$. Our analysis will focus on the marginal distribution of individual  resampled, mutated, or coupled particles. We denote these marginal distributions as 
\begin{align*}
\tilde{\mu}_{v}(B) & := \Prob(\tilde{X}^{i}_{v} \in B)\\
\hat{\mu}_{v}(B) & := \Prob(X^{i}_{v} \in B) \\
\bar{\mu}_{v}(B) & := \Prob(\bar{X}^{i}_{v} \in B),
\end{align*}
which are identical for $i=1,\ldots,N$ by symmetry. Similarly, conditional distributions are denoted by $\tilde{\mu}(B \mid H) = \Prob(\tilde{X}^{i}_{v} \in B \mid H)$ for $B \subset \mathcal{X}$ and $H \in \mathcal{B}^{N(3V + 1)}$. In addition, we will use the convenient shorthand  $\tilde{\mu}_{v|A_{j}}$, $\hat{\mu}_{v|A_{j}}$, and $\bar{\mu}_{v|A_{j}}$ to denote the corresponding restricted distributions, e.g. $\tilde{\mu}_{v|A_{j}}(B) = \tilde{\mu}_{v}(B \mid \tilde{X}_{v}^{i} \in A_{j})$, and the mixed notation $\tilde{\mu}_{v|A_{j}}(B \mid H) = \Prob(\tilde{X}^{i}_{v} \in B \mid \tilde{X}^{i}_{v} \in A_{j}, H)$ for restrictions of conditional distributions. We note that $\bar{\mu}_{v} \neq \mu_{v}$ in general (see Appendix~\ref{appendix_a} for details). Finally, we define the filtration $\mathcal{F}_{v} := \sigma(\{X^{1:N}_{k}, \tilde{X}^{1:N}_{k}\}^{v}_{k=1})$ with $\mathcal{F}_{0} = \sigma(X^{1:N}_{0})$ and denote the probability of resampling into $A_{j}$ as
\begin{align*}
    \hat{P}^{j}_{v} := \tilde{\mu}_{v}(A_{j} \mid \mathcal{F}_{v-1}) = \frac{\hat{w}^{j}_{v}}{\sum^{p}_{r=1} \hat{w}^{r}_{v}},
\end{align*}
where $\hat{w}^{j}_{v}$ is defined in Section~\ref{smc}. That is, $\{\hat{P}^{j}_{v}\}$ is a stochastic processes adapted to the filtration $\{\mathcal{F}_{v}\}$.  

\section{Approach}
Given some $\Prob$-measurable function $\abs{f} \leq 1$, our goal is to show that the SMC estimator $\hat{f} = \frac{1}{N}\sum_{i=1}^N f(X_V^i)$
approximates $E_\pi[f]$ to high accuracy with high probability. The challenge is to do so without requiring $X_V^i \sim \pi$. Observe that 
\begin{align*} \E_{\pi}[f] = \sum^{p}_{j=1} \E_{\pi_{\mid A_{j}}}[f] \pi(A_{j})
\quad \text{ while } \quad
%
\E[\hat{f}] =
\sum^{p}_{j=1} \E_{\hat{\mu}_{V \mid A_{j}}}[\hat{f}_{j}] \, \hat{\mu}_{V}(A_{j})
\end{align*}
where  $\hat{f}_{j} = \frac{1}{N}\sum^{N}_{i=1} f(X^{i}_{V}) \mathbbm{1}_{A_{j}}(X^{i}_{V})$. 
%
Consequently, we must control two sources of error: the deviation of $\hat{f}$
from its expectation $|\hat{f} - \E[\hat{f}]|$, and the error in the expectation itself (bias) $|\E[\hat{f}] - \E_\pi[f]|$. 
Controlling the latter source of error is the more delicate of the two. To do so, we will define a sequence of events conditional upon which $\hat{\mu}_{v\mid A_j} \approx \mu_{v\mid A_j}$
and $\hat{\mu}_v(A_j) \approx \mu_v(A_j)$, and then show that this sequence of events occurs with sufficiently high probability. This requires choosing $N$ sufficiently large to sequentially control the bias introduced by resampling.  Previously, \citet{marion}  control this by constructing a coupling to \textit{i.i.d.} particles and invoking concentration properties thereof. The main difficulty here is that, in the absence of assumptions of global mixing, the coupling construction provides particles which are neither independent nor have the correct distribution ($\mathcal{L}(\bar{X}^{i}_{v}) \neq \mu_{v}$). As a result, establishing concentration properties involving $\bar{X}^{1:N}_{v}$ (and consequently $X^{1:N}_{v}$) is non-trivial.

For $j = 1,\ldots,p$ and $v = 1,\ldots,V$, define the following events on the probability space $(\mathcal{X}^{N(3V+1)},\mathcal{B}^{N(3V + 1)}, \Prob)$:
\begin{align}
\mathcal{C}^{j}_{1}(i) &= \left\{ \abs{\bar{w}^{j}_{1} - \frac{z_{1}}{z_{0}} \mu_{1}(A_{j})} \leq \lambda \cdot \frac{z_{1}}{z_{0}} \mu_{1}(A_{j}) \right\} \label{initial} \\
\mathcal{C}^{j}_{v}(i) &= \left\{\abs{\bar{w}^{j}_{v} - \E[\bar{w}^{j}_{v} \mid \mathcal{F}_{v-2}]} \leq \lambda \cdot \E[\bar{w}^{j}_{v} \mid \mathcal{F}_{v-2}] \right\}
\label{conc} \text{ for } v \geq 2 \\
\mathcal{C}_{v}(ii) &= \left\{\bar{X}^{1:N}_{v} = X^{1:N}_{v} \right\}, \label{coupling_event}
\end{align}
where $\bar{w}^{j}_{v} = N^{-1} \sum^{N}_{i=1} w_{v}(\bar{X}^{i}_{v-1})\mathbbm{1}_{A_j}(\bar{X}_{v-1}^{i})$, and $\lambda = \frac{\epsilon}{24V}$ where $\epsilon \in (0,\frac{1}{2})$ is the error tolerance chosen by the user in the application of  Theorem~\ref{thm:maintheorem}. The event $\mathcal{C}_{v}^j(i)$ represents the successful approximation by $\bar{w}^{j}_{k}$ of its conditional mean to within desired relative error bound $\lambda$. 
Note that $ \E[\bar{w}^{j}_{v} | \mathcal{F}_{v-2}]$ is a random variable; the event 
$\mathcal{C}^{j}_{v}(i)$ 
corresponds to $\bar{w}^{j}_{v}$ being `sandwiched' between two other random variables. $\mathcal{C}_v(ii)$ is the simultaneous coupling event for all particles at step $k$. 
Our approach is to show that \textit{all} of the events $\mathcal{C}_v^j(i)$, $\mathcal{C}_v^j(ii)$, $j=1,\ldots,p$, $k=1,\ldots,V$, 
occur simultaneously with high probability, and that conditional on their simultaneous occurrence we have  $\hat{\mu}_{v\mid A_j} \approx \mu_{v\mid A_j}$ and $\hat{\mu}_v(A_j) \approx \mu_v(A_j)$ for all $v=1,\ldots,V$.

In Section~\ref{sec: results cond. on E}, we show that simultaneous occurrence of these events ensures $\hat{\mu}_{v\mid A_j} \approx \mu_{v\mid A_j}$ and $\hat{\mu}_v(A_j) \approx \mu_v(A_j)$, and hence that $\sum^{p}_{r=1}\E[\hat{w}^{r}_{v}]$ closely approximates $z_{v+1} / z_{v}$.
In Section~\ref{sec: induction argument} we show that all of the events $\mathcal{C}_v^j(i)$, $\mathcal{C}_v^j(ii)$, $j=1,\ldots,p$, $k=1,\ldots,V$, 
occur simultaneously with high probability. We will do so by induction on $k$. In particular, 
let $ \mathcal{C}^{1:p}_{v}(i) = \cap^{p}_{r=1}  \mathcal{C}^{r}_{v}(i)$. For each $V$ and for a chosen $\lambda$ we will show that the probability of the following event is large:
\begin{align*}
\mathcal{E}_{v} := 
\cap^{v}_{k=1}
\{ \mathcal{C}^{1:p}_{k}(i) \cap \mathcal{C}_{k}(ii) \} 
\end{align*}
$\mathcal{E}_v$ is the joint event that \textit{all} (\ref{initial}), (\ref{conc}), and (\ref{coupling_event}) occur for steps $1,\ldots,v$. It will also be convenient to define the events $\mathcal{G}_{v} := \mathcal{E}_{v-1} \cap \mathcal{C}^{1:p}_{v}(i)$ for $v \geq 2$ and $\mathcal{G}_{1} = \mathcal{C}^{1:p}_{1}(i)$ so that $\mathcal{E}_{v} = \mathcal{G}_{v} \cap  \mathcal{C}_{v}(ii)$.

\section{Results}\label{results}
Our main result for the algorithm in Section \ref{smc} is stated in the following Theorem. Supporting lemmas used in the proof of  Theorem~\ref{thm:maintheorem} are given in Section~\ref{sec: results cond. on E} and Section~\ref{sec: induction argument}.
\begin{theorem}\label{thm:maintheorem}
Suppose $|f| \leq 1$ and (\ref{assumption: density ratios}) holds. Let $\epsilon \in (0,\frac{1}{2})$ and choose 
\begin{enumerate}
    \item  $ N > \frac{1}{\epsilon^{2}} \cdot  \max\left\{3456 \cdot  \left(\frac{VWZ}{\mu^{\star}} \right)^{2} \cdot \log(\frac{64 Vp}{ \mu^{\star}}), \ p^{2} \cdot  \log(1024  p^{2})   \right\}$
    \item $t > \max\limits_{j=1,\ldots,p} \max\limits_{v=1,\ldots,V} \tau_{v,j}( \frac{\mu^{\star}}{16  N  V}, 7)$,
\end{enumerate}
Then with probability at least $3/4$,
\begin{align*}
     \left|\frac{1}{N}\sum^{N}_{i=1}f(X^{i}_{V})  - \frac{(1+\frac{\epsilon^{2}}{16})}{(1-\frac{\epsilon^{2}}{16})}\E_{\pi}[f]\right| \leq \epsilon.
\end{align*}
\end{theorem}
\noindent 
Note that Theorem~\ref{thm:maintheorem} provides bounds on the total number of Markov transition steps $NVt$ required to approximate $\E_{\pi}[f]$ for an \textit{arbitrary} bounded function $f(x)$. For specific choices of $f$, improved bounds may be available. 

\subsection{Results Over Restricted Space}
\label{sec: results cond. on E}

We establish three key lemmas.  The first shows that $\hat{P}^{j}_{v}$ closely approximates $\mu_{v}(A_{j})$ over $\mathcal{G}_{v} \subset \mathcal{E}_{v}$. The second shows that $\sum^{p}_{r=1}\hat{w}^{r}_{v}$ closely approximates $z_{v+1} / z_{v}$ conditional on $\mathcal{G}_{v}$.
The third shows that the distribution $\tilde{\mu}_{v}(\cdot \mid \mathcal{G}_{v})$ is locally 7-warm (see Definition~\ref{def:locally warm}) with respect to $\mathcal{A}$. Note that all of these results describe what happens \textit{conditional} on $\mathcal{G}_{v}$ holding. In Section~\ref{sec: induction argument}, we use the results developed here to show also that $\mathcal{G}_{v}$ occurs with high probability. The results in these two sections are then combined to obtain the proof of  Theorem~\ref{thm:maintheorem}.
\begin{lemma}
\label{controlling_resampling}
Let $\phi(\lambda) = \phi = \frac{1+\lambda}{1-\lambda}$. For $v = 1,\ldots,V$,
\begin{align*}
   \Prob\left(\phi^{-v} \mu_{v}(A_{j})  \leq 
   \hat{P}_v^j
   \leq \phi^{v} \mu_{v}(A_{j})  \mid \mathcal{G}_{v} \right) = 1.
   \end{align*}
   In particular, since $\tilde{\mu}(A_j \mid \Gv) = \Prob(\tilde{X}^i_{v} \in A_{j} \mid \Gv) = 
    \E\left[\hat{P}^j_v \mid \Gv \right]$ we have:
    \begin{align*}
        \phi^{-v} \mu_{v}(A_{j}) \leq \tilde{\mu}_v(A_{j} \mid \Gv) \leq \phi^v \mu_v(A_j).
    \end{align*}
\end{lemma}
\begin{proof}(Lemma~\ref{controlling_resampling})
First notice that given $\mathcal{G}_v$ we have 
\begin{align}
\hat{P}_v^j = \frac{\hat{w}^{j}_{v}}{\sum^{p}_{r=1} \hat{w}^{r}_{v}} 
=  \frac{\bar{w}^{j}_{v}}{\sum^{p}_{r=1} \bar{w}^{r}_{v}} 
  & \leq \frac{(1+\lambda)\E[\bar{w}^{j}_{v}\mid \mathcal{F}_{v-2}]}{(1-\lambda)\sum^{p}_{r=1} \E[\bar{w}^{r}_{v} \mid \mathcal{F}_{v-2}]}
  \nonumber \\ 
  & = \phi  \frac{\E[\bar{w}^{j}_{v} \mid \mathcal{F}_{v-2}]}{\sum^{p}_{r=1} \E[\bar{w}^{r}_{v} \mid \mathcal{F}_{v-2}]}
\label{Eqn:prop1ineq1}
\end{align}
where the second equality follows by recalling that $\mathcal{C}_{v-1}(ii) = \{X_{v-1}^{1:N} = \bar{X}^{1:N}_{v-1} \}$ and  $\mathcal{C}_{v-1}(ii) \subset \mathcal{G}_{v}$, and the first inequality by recalling that $\mathcal{C}^{1:p}_{v}(i) \subset \mathcal{G}_{v}$
and $\mathcal{C}^{1:p}_{v}(i) =  \{(1-\lambda)\E[\bar{w}^{j}_{v} | \mathcal{F}_{v-2}] \leq \bar{w}^{j}_{v} \leq (1+\lambda)\E[\bar{w}^{j}_{v} | \mathcal{F}_{v-2}]  \}$. Lemma~\ref{cond_exp_equality} in Appendix~\ref{appendix_a} shows that  
\begin{align}
     \E[\bar{w}^{j}_{v} | \mathcal{F}_{v-2}] = \frac{z_{v}}{z_{v-1}} \cdot \frac{\mu_{v}(A_{j})}{\mu_{v-1}(A_{j})} \cdot \hat{P}^{j}_{v-1}
\label{Eqn:prop1ineq2}
\end{align}
Intuitively, this means that $\E[\bar{w}^{j}_{v} | \mathcal{F}_{v-2}]$ is the ratio of the normalizing constants of the restricted distributions $\mu_{v|A_{j}}$ and $\mu_{v-1|A_{j}}$
but weighted by the probability that a particle is resampled into $A_{j}$. We can think of the filtration as being summarized solely through $\hat{P}^{j}_{v-1}$ due to the conditional independence properties of the constructed particles (see Lemma~\ref{cond_ind} in Appendix~\ref{appendix_a}). 
Combining (\ref{Eqn:prop1ineq1}) and (\ref{Eqn:prop1ineq2}) and canceling, we have 
\begin{align}
   \frac{\hat{w}^{j}_{v}}{\sum^{p}_{r=1} \hat{w}^{r}_{v}} 
\leq  
\phi\frac{ \frac{\mu_{v}(A_{j})}{\mu_{v-1}(A_{j})}  \hat{P}^{j}_{v-1}}{\sum^{p}_{r=1}   \frac{\mu_{v}(A_{r})}{\mu_{v-1}(A_{r})}  \hat{P}^{r}_{v-1}}
= \phi\frac{ \frac{\mu_{v}(A_{j})}{\mu_{v-1}(A_{j})}   \hat{w}^{j}_{v-1}}{\sum^{p}_{r=1}   \frac{\mu_{v}(A_{r})}{\mu_{v-1}(A_{r})}  \hat{w}^{r}_{v-1}}
\label{Eqn:prop1ineq3}
\end{align}
Barring the constant terms, the right hand side of (\ref{Eqn:prop1ineq3}) is the same as the left hand side except the index has been decremented by one. 
Recall 
that $\mathcal{G}_{v}$
%
ensures that the events $\mathcal{C}_v^{1:p}(i)$ and $\mathcal{C}_{v-1}(ii)$ which we appealed to for step $v$ also hold for all previous steps $0,\ldots,v-1$. Thus given $\mathcal{G}_v$ the argument above can be repeated recursively to give 
\begin{align*}
   \frac{\hat{w}^{j}_{v}}{\sum^{p}_{r=1} \hat{w}^{r}_{v}} 
   \leq 
\phi\frac{ \frac{\mu_{v}(A_{j})}{\mu_{v-1}(A_{j})}   \hat{w}^{j}_{v-1}}{\sum^{p}_{r=1}   \frac{\mu_{v}(A_{r})}{\mu_{v-1}(A_{r})}  \hat{w}^{r}_{v-1}}
& \leq
    \phi^{2}\frac{ \frac{\mu_{v}(A_{j})}{\mu_{v-2}(A_{j})}  \hat{w}^{j}_{v-2}}{\sum^{p}_{r=1}   \frac{\mu_{v}(A_{r})}{\mu_{v-2}(A_{r})}  \hat{w}^{r}_{v-2}}
    \\
 \ldots    
    & \leq \phi^{v-1}\frac{ \frac{\mu_{v}(A_{j})}{\mu_{1}(A_{j})}  \hat{w}^{j}_{1}}{\sum^{p}_{r=1}   \frac{\mu_{v}(A_{r})}{\mu_{1}(A_{r})}  \hat{w}^{r}_{1}}
\end{align*}
Recall that $\mathcal{C}^{1:p}_{0}(i) \subset \mathcal{G}_{v}$ is the event that $\hat{w}^{j}_{1}$ 
closely approximates the normalizing constant of the restricted distribution $\mu_{1|A_{j}}$; plugging this in gives
\begin{align*}
\hat{P}_v^j =  \frac{\hat{w}^{j}_{v}}{\sum^{p}_{r=1} \hat{w}^{r}_{v}} 
  \leq
\phi^{v}\frac{ \frac{\mu_{v}(A_{j})}{\mu_{1}(A_{j})} \frac{z_{1} \mu_{1}(A_{j})}{z_{0}}}{\sum^{p}_{r=1}   \frac{\mu_{v}(A_{r})}{\mu_{1}(A_{r})}  \frac{z_{1} \mu_{1}(A_{r})}{z_{0}}}
    = \phi^{v} \mu_{v}(A_{j})
\end{align*}
Repeating the argument in the other direction (with $\phi^{-1}$ in place of $\phi$) yields the result. 
\end{proof}

\begin{remark}
Note that the same argument implies
\begin{align*}
   \Prob\left(\phi^{-v} \mu_{v}(A_{j})  \leq 
   \hat{P}_v^j
   \leq \phi^{v} \mu_{v}(A_{j})  \mid \mathcal{G}_{v'} \right) = 1
\end{align*}
for any $v' \geq v$. To see this, notice that $\mathcal{G}_{v} \subset \mathcal{G}_{v^{\prime}}$ and the above argument holds for any set containing $\mathcal{G}_{v}$.
\end{remark}
To establish local warmness of $\tilde{\mu}_v$, we will use the following lemma which shows that $\sum_{j}\hat{w}^{j}_{v}$ closely approximates $z_{v} / z_{v-1}$ over $\mathcal{G}_{v}$. This establishes that, within the set $\mathcal{G}_{v}$, the algorithm given in Section~\ref{smc} can be used to approximate $z_{v}$ for any $v=1,\ldots,V$ when the initial $z_{0}$ is known.
\begin{lemma}
\label{lemma: normalizing constant bound}
For $v=1,\ldots,V$ we have
\begin{align*}
    \Prob\left( \frac{1}{\phi^{v}} \frac{z_{v}}{z_{v-1}} \leq \sum^{p}_{r=1}\hat{w}^{r}_{v} \leq \phi^{v} \frac{z_{v}}{z_{v-1}} \Bigm| \mathcal{G}_{v} \right) = 1
\end{align*}
\end{lemma}
\begin{proof}
Note that $\mathcal{C}_{v-1}(ii) \subset \Gv$ and so $\hat{w}^{j}_{v} = \bar{w}^{j}_{v}$ for $j=1,\ldots,p$ given $\Gv$. We then have
\begin{align*}
\sum_{j=1}^p \bar{w}_v^j \geq \frac{z_v}{z_{v-1}}\sum_{j=1}^p (1-\lambda) \frac{\mu_v(A_j)}{\mu_{v-1}(A_j)}\hat{P}_{v-1}^j
&\geq \frac{(1-\lambda)}{\phi^{v-1}}\frac{z_v}{z_{v-1}}\sum_{j=1}^p \mu_v(A_j) \\
&\geq \frac{1}{\phi^v}\cdot\frac{z_v}{z_{v-1}}
\end{align*}
The first inequality follows since $\mathcal{C}^{1:p}_{v}(i) \subset \Gv$ and by  Lemma~\ref{cond_exp_equality} (see Appendix~\ref{appendix_a}). The second inequality follows by Lemma~\ref{controlling_resampling}. An identical argument can be repeated in the other direction showing $\sum_{j=1}^p \bar{w}_v^j \leq \phi^{v} \frac{z_{v}}{z_{v-1}}$.
\end{proof}
We now show that $\tilde{\mu}_{v}(\cdot \mid \mathcal{G}_{v})$ is a locally 7-warm distribution with respect to $\mu_{v}$ and $\mathcal{A}$. To do this, we first state an inequality regarding $\phi^{v}$.
\begin{proposition}
\label{resampling_error}
Let $\phi(\lambda) = \phi = \frac{1+\lambda}{1-\lambda}$ and recall $\lambda \leq \frac{\epsilon}{6V}$
. Then $\phi^{v} <\phi^{2v} < \frac{1 + \epsilon}{1-\epsilon}$ for $v = 1,\ldots,V$,
\end{proposition}
\begin{proof}
\begin{align*}
    \phi^{2v} = \left(\frac{1 + \frac{\epsilon}{6V}}{1 - \frac{\epsilon}{6V}}\right)^{2v} \leq \frac{e^{\frac{\epsilon}{3}}}{1-\frac{\epsilon}{3}} \leq \frac{1 + \frac{\epsilon}{3} + \frac{\epsilon^{2}}{9}}{1 - \frac{\epsilon}{3}} \leq \frac{1 + \epsilon}{1-\epsilon},
\end{align*}
where we used the elementary inequalities $1 + x \leq e^{x}$, $(1 + x)^{p} \geq 1 + xp$ for $x > -1$ and $p > 1$, and $e^{x} \leq 1 + x + x^{2}$ for $x \in (0,1)$.
\end{proof}
\noindent Proposition \ref{resampling_error} shows that taking $\lambda$ to be $\mathcal{O}(V^{-1})$ suffices to ensure the bound in Lemma~\ref{controlling_resampling} is uniform over $v = 1,\ldots,V$.
\begin{lemma}
\label{warmness}
Suppose $\Prob(\mathcal{G}^c_v) \leq \frac{\mu^{\star}}{4}$ and recall $\lambda = \frac{\epsilon}{24V} \leq \frac{1}{6V}$. Then $\tilde{\mu}(\cdot \mid \Gv)$ is locally 7-warm with respect to $\mu_{v}$ for partition $\mathcal{A}$; that is,
\begin{align*}
     \sup_{B \subset A_{j}}\frac{\tilde{\mu}_{v\mid A_{j}}(B \mid \mathcal{G}_{v})}{\mu_{v\mid A_{j}}(B)} < 7, \text{ for }j=1,\ldots,p.
\end{align*}
\end{lemma}
\begin{proof}
Expand 
\begin{align}
\frac{\tilde{\mu}_{v\mid A_{j}}(B \mid \Gv)}{\mu_{v\mid A_{j}}(B)} =
\frac{\tilde{\mu}_v(B \cap A_j \mid  \Gv)\mu_v(A_j)}{\tilde{\mu}_v(A_j \mid \Gv)\mu_v(B \cap A_j)}
\label{Eqn:RatioExpansion}
\end{align}
and consider first the term $\tilde{\mu}_v(B \cap A_j \mid  \Gv) = \Prob(\tilde{X}^{i}_{v} \in B \cap A_{j} \mid \Gv)$.  We have
\begin{align*}
\Prob(\tilde{X}^{i}_{v} \in B \cap A_{j} \mid \Gv) &=
\E\left[
\frac{\frac{1}{N}\sum^{N}_{n=1} w_{v}(\bar{X}^{n}_{v-1}) \mathbbm{1}_{A_{j} \cap B}(\bar{X}^{n}_{v-1})}
{\sum^{p}_{r=1} \bar{w}^{r}_{v}} \Bigm| \Gv
\right] \\
%
%
%
%
%
&\leq 
\phi^{v}
\E\left[
\frac{\mu_{v}(\bar{X}^{i}_{v-1})}{\mu_{v-1}(\bar{X}^{i}_{v-1})} \mathbbm{1}_{A_{j} \cap B}(\bar{X}^{i}_{v-1})
 \Bigm| \Gv
\right]  \\
&=
\frac{\phi^{v}}{\Prob(\mathcal{G}_{v})}
\E\left[
\frac{\mu_{v}(\bar{X}^{i}_{v-1})}{\mu_{v-1}(\bar{X}^{i}_{v-1})} \mathbbm{1}_{A_{j} \cap B}(\bar{X}^{i}_{v-1}) \GvInd
\right]  \\
%
%
%
%
&  \leq  \frac{\phi^{v}}{\Prob(\mathcal{G}_{v})} \E\left[\frac{\mu_{v}(\bar{X}^{i}_{v-1}) \mathbbm{1}_{A_{j} \cap B}(\bar{X}^{i}_{v-1} )}{\mu_{v-1}(\bar{X}^{i}_{v-1})}  \Bigm| \bar{X}^{i}_{v-1} \in A_{j} \right] 
\bar{\mu}_{v-1}(A_{j})
\\
%
%
%
%
%
&=  \frac{\phi^{v}}{\Prob(\mathcal{G}_{v})}  \mu_{v}(A_{j} \cap B)  \frac{\tilde{\mu}_{v-1}(A_{j})}{\mu_{v-1}(A_{j})}
\end{align*}
The first inequality follows by Lemma~\ref{lemma: normalizing constant bound}. The second inequality follows from non-negativity of the integrand and from the law of total probability.  The latter is required because $\bar{\mu}_{v-1} \neq \mu_{v-1}$ but
$\Prob(\bar{X}^{i}_{v-1} \in B' \mid \bar{X}^{i}_{v-1} \in A_{j}) = \mu_{v-1|A_{j}}(B{'})$ for $B{'} \subset \mathcal{X}$ (see Appendix~\ref{appendix_a}).  This fact is also used in the final line, along with the fact that $\bar{\mu}_{v}(A_{j}) = \tilde{\mu}_{v}(A_{j})$ holds for all $i,j,v$ by construction (see Appendix~\ref{appendix_a}). As a result, from (\ref{Eqn:RatioExpansion}) we have:
\begin{align*}
\frac{\tilde{\mu}_{v\mid A_{j}}(B \mid \Gv)}{\mu_{v\mid A_{j}}(B)} \leq
\frac{\phi^{v}}{\Prob(\Gv)}   \frac{\Prob(\tilde{X}^{i}_{v-1} \in A_{j})}{\mu_{v-1}(A_{j})} \cdot
\frac{\mu_v(A_j)}{\tilde{\mu}_v(A_j \mid \Gv)}.
\end{align*}
Now by Lemma~\ref{controlling_resampling}, $\tilde{\mu}_v(A_j \mid \Gv) \geq \phi^{-v} \mu_{v}(A_{j})$
and therefore
$\frac{\mu_{v}(A_{j})}{\tilde{\mu}_v(A_j \mid \Gv)} \leq \phi^{v}$ provides a lower bound on the relative marginal probability of undersampling partition elements. From Lemma~\ref{controlling_resampling} we also  have
\begin{align*}
    \Prob(\tilde{X}^i_{v-1} \in A_{j}) \,  = \tilde{\mu}_{v-1}(A_j) \leq \phi^{v-1}\mu_{v-1}(A_{j}) \Prob(\mathcal{G}_{v}) + \Prob(\mathcal{G}^{c}_{v}).
\end{align*}
Our assumptions on $\lambda$ and $\Prob(\mathcal{G}^{c}_{v})$ then yield
\begin{align*}
 \sup_{B \subset A_j}\frac{\tilde{\mu}_{v \mid A_j}(B \mid \Gv)}{\mu_{v \mid A_j}(B)} 
&\leq 
\frac{\phi^{2v}}{\Prob(\Gv)\mu_{v-1}(A_{j})} \Prob(\tilde{X}^i_{v-1} \in A_{j})\\
&\leq  \frac{\phi^{2v}}{\Prob(\Gv)}\left[ \frac{ \phi^{v-1}\mu_{v-1}(A_{j}) \Prob(\Gv) + \Prob(\Gv^c)}{\mu_{v-1}(A_{j})} \right] \\
&\leq \phi^{3v-1} + \frac{\phi^{2v}}{\Prob(\Gv)} \\ &\leq 4 \, e^{\frac{1}{2}} \,< \, 7 
\end{align*}
The third inequality follows since 
$\mu^{\star} \leq \mu_{v-1}(A_{j})$. The final inequality is obtained by repeating the argument of Proposition~\ref{resampling_error}
for the choice of $\lambda$ to get $\phi^{jv} < \frac{6e^{\frac{j}{6}}}{6-j}$, and using $\Prob(\mathcal{G}_{v}) > 3/4$.
\end{proof}

\subsection{Induction Argument}\label{sec: induction argument}

The remainder of the analysis is to show that $\mathcal{G}_{v}$ holds with high probability. We will do so inductively by first establishing that $\Prob(\mathcal{E}_v \mid \mathcal{E}_{v-1})$  is sufficiently large for all $v$.
\begin{lemma}
\label{resampling_bound}
Let $0 < \delta  < 1$ and suppose $N > \frac{2W^{2}Z^{2} \phi^{2(v-1)} \log(\frac{4p}{\delta})}{ \lambda^{2} (\mu^{\star})^{2}}$. Then
\begin{align*}
           \Prob(\mathcal{C}^{1:p}_{v}(i) \mid  \mathcal{E}_{v-1}) \geq 1 - \frac{ \delta}{\Prob( \mathcal{E}_{v-1})}
        \end{align*}
\end{lemma}
\begin{proof}
Note that $\mathcal{G}_{v-1} \subset \mathcal{E}_{v-1}$. By combining Lemma~\ref{cond_exp_equality} (Appendix~\ref{appendix_a}) and Lemma~\ref{controlling_resampling}, we have over $\mathcal{G}_{v-1}$,
\begin{align*}
    \E[\bar{w}^{j}_{v} \mid \mathcal{F}_{v-2}] = \frac{z_{v}}{z_{v-1}} \cdot \frac{\mu_{v}(A_{j})}{\mu_{v-1}(A_{j})} \cdot \hat{P}^{j}_{v-1} \geq \frac{z_{v}}{z_{v-1}} \cdot \phi^{-(v-1)} \cdot \mu_{v}(A_{j})
\end{align*}
Then 
\begin{align*}
\Prob((\mathcal{C}^{j}_{v}(i))^{c} \mid \mathcal{E}_{v-1}) 
    = &\Prob\left(\abs{\bar{w}^{j}_{v} - \E[\bar{w}^{j}_{v} \mid \mathcal{F}_{v-2}]} > \lambda \E[\bar{w}^{j}_{v} \mid \mathcal{F}_{v-2}] \mid   \mathcal{E}_{v-1} \right) \\
    \leq &\Prob\left(\abs{\bar{w}^{j}_{v} - \E[\bar{w}^{j}_{v} \mid \mathcal{F}_{v-2}]} > \lambda \frac{z_{v}}{z_{v-1}}  \phi^{-(v-1)} \mu_{v}(A_{j}) \mid \mathcal{E}_{v-1} \right) \\
    \leq &\frac{1}{\Prob( \mathcal{E}_{v-1})}\Prob\left(\abs{\bar{w}^{j}_{v} - \E[\bar{w}^{j}_{v} \mid \mathcal{F}_{v-2}]} > \lambda \frac{z_{v}}{z_{v-1}} \phi^{-(v-1)}  \mu_{v}(A_{j})   \right) \\
    \leq &\frac{ \delta}{p\Prob( \mathcal{E}_{v-1})}
\end{align*}
where the last line
in Appendix~\ref{appendix_a} with $N > \frac{2W^{2}Z^{2} \phi^{2(v-1)} \log(\frac{4p}{\delta})}{\lambda^{2} (\mu^{\star})^{2}}$. By the union bound,
\begin{align*}
\Prob(\cup^{p}_{j=1} (\mathcal{C}^{j}_{v}(i))^{c} \mid   \mathcal{E}_{v-1}) \leq \frac{ \delta}{\Prob(\mathcal{E}_{v-1})}
\end{align*}
DeMorgan's law gives the result.
\end{proof}
To complete the inductive step, it remains to show that 
$P(\mathcal{C}_v(ii) \mid \mathcal{C}_v^{1:p}, \mathcal{E}_{v-1})$ is also large.  To do so, we use the local warmness result established in Lemma~\ref{warmness}.


%
%
\begin{lemma}
\label{coupling}
Suppose $\Prob(\mathcal{G}^{c}_{v}) \leq \frac{\mu^{\star}}{4}$ and recall $\lambda \leq \frac{1}{6V}$. Let
\begin{align*}
    t > \max_{v=1,\ldots,V} \max_{j=1,\ldots,p} \tau_{v,j}\left(\frac{\delta}{N}, 7\right)
\end{align*}
Then $\Prob(\mathcal{C}_{v}(ii) \mid \mathcal{G}_{v}) \geq 1 - \delta$.
 \end{lemma}
 \begin{proof}
Since $\hat{\mu}_{v}(A_{j})  = \tilde{\mu}_{v}(A_{j})$  for all $v,j$ (see Appendix~\ref{appendix_a}), we have for $j = 1,\ldots,p$ and $B \subset \mathcal{X}$,
\begin{align*}
    \hat{\mu}_{v|A_{j}}( B \mid \mathcal{G}_{v})
    &= \frac{\Prob(X_v^i \in B, \tilde{X}_v^i \in A_j \mid \Gv)}{\Prob(\tilde{X}_v^i \in A_j \mid \Gv)} \\
    &= \frac{1}{\tilde{\mu}_{v}(A_{j} \mid \Gv)}\int_{A_{j}} \tilde{\mu}_{v}(dx \mid \Gv) K_{v|A_{j}}^{t}(x,B)  \\
     &= \int_{\mathcal{X}} \tilde{\mu}_{v|A_{j}}(dx \mid \Gv)  K_{v|A_{j}}^{t}(x,B)
\end{align*}
By Lemma \ref{warmness}, $\tilde{\mu}_{v|A_{j}}(\cdot \mid \Gv) \in \mathcal{P}_{7}(\mu_{v|A_{j}})$. Hence, by our choice of $t$,
\begin{align*}
    \norm{\hat{\mu}_{v|A_{j}}( \cdot \mid \Gv) - \mu_{v| A_{j}}(\cdot)}_{\text{TV}} \leq \frac{\delta}{N} \text{ for }j=1,\ldots,p
\end{align*}
It follows that we can construct $(X^{i}_{v}, \bar{X}^{i}_{v})$ conditional on $\tilde{X}^{i}_{v}$ independently for $i = 1,\ldots,N$ such that
\begin{align*}
    \Prob(X^{i}_{v} \neq \bar{X}^{i}_{v} \mid \tilde{X}^{i}_{v} \in A_{j}, \Gv) \leq \frac{\delta}{N},  \text{ for }j=1,\ldots,p
\end{align*}
(see Appendix~\ref{appendix_a} for details). It follows that $\Prob(X^{i}_{v} \neq \bar{X}^{i}_{v} \mid \Gv) \leq \frac{\delta}{N}$ by the law of total probability. Finally, taking a union bound gives us the stated result: 
\begin{align*}
   \Prob(\mathcal{C}_{v}(ii) \mid \Gv) =  \Prob(X^{1:N}_{v} \neq \bar{X}^{1:N}_{v} \mid \Gv) &\leq \delta
\end{align*}
\end{proof}
Lemma \ref{resampling_bound} and Lemma \ref{coupling} are enough to establish that $\Prob(\mathcal{E}_{v} \mid \mathcal{E}_{v-1})$ can be made sufficiently large, which gives us the following result.
\begin{lemma}
\label{OneStepError}
Let $\delta \in (0,\frac{1}{4})$ and $\epsilon \in (0,\frac{1}{2})$. Choose 
\begin{enumerate}
    \item  $ N > 3456 \cdot \left(\frac{VWZ}{\epsilon \mu^{\star} }\right)^{2} \cdot  \log\left(\frac{8Vp}{\delta \mu^{\star}}\right)$
    \item $t > \max\limits_{j=1,\ldots,p} \max\limits_{v=1,\ldots,V} \tau_{v,j}(\frac{\delta \mu^{\star}}{2NV}, 7)$
\end{enumerate}
Then
\begin{align*}
    \Prob(\mathcal{E}_{V}) \geq 1 - \delta \mu^{\star} 
\end{align*}
\end{lemma}
\begin{proof}
First, suppose that $\Prob(\mathcal{G}_{k}^c) \leq \frac{\mu^{\star}}{4}$ for $k = 1,\ldots,v$ so that we can directly apply Lemma \ref{warmness} and Lemma \ref{coupling} for steps $k = 1,\ldots,v$. Let $\delta^{\prime} = \frac{\delta \mu^{\star}}{2V}$. Then by the definition of $\mathcal{E}_{v}$ and our choice of $N$ and $t$,
\begin{align*}
    \Prob(\mathcal{E}_{v}) &= \Prob(\mathcal{C}^{1:p}_{v}(i)\, \cap \, \mathcal{C}_{v}(ii) \, \cap\,  \mathcal{E}_{v-1}) \\
    &=\Prob( \mathcal{C}_{v}(ii) \mid \mathcal{C}^{1:p}_{v}(i) \, \cap \, \mathcal{E}_{v-1})\Prob(\mathcal{C}^{1:p}_{v}(i) \mid \mathcal{E}_{v-1}) \Prob(\mathcal{E}_{v-1}) \\
    &\geq (1-\delta^{\prime})\Big(1 - \frac{\delta^{\prime}}{\Prob(\mathcal{E}_{v-1})}\Big) \Prob(\mathcal{E}_{v-1}) \\
    &\geq \Prob(\mathcal{E}_{v-1}) - 2\delta^{\prime}
\end{align*}
The first inequality follows by Lemmas~\ref{resampling_bound} and 
Lemma~\ref{coupling}. (Note that the choice of $N$ here satisfies the bound required to apply Lemma~\ref{resampling_bound} since 
$\phi^{2V} < 3$
for $\epsilon < \frac{1}{2}$ by Proposition \ref{resampling_error}).
Since we have assumed $\Prob(\mathcal{G}_{k}^c) \leq \frac{\mu^{\star}}{4}$ for $k = 1,\ldots,v$, we can apply this argument recursively to obtain
\begin{align}
    \Prob(\mathcal{E}_{v}) \geq \Prob(\mathcal{E}_{1}) - 2(v-1)\delta^{\prime} &=  \Prob(\mathcal{C}_{1}(ii) \mid \mathcal{G}_{1})\Prob(\mathcal{G}_{1}) - 2(v-1)\delta^{\prime} \nonumber \\
    &\geq (1 - \delta^{\prime})(1-\delta^{\prime}) - 2(v-1)\delta^{\prime} \nonumber \\
    &\geq 1 - 2v\delta^{\prime}\nonumber \\
    &= 1 - \delta  \mu^{\star}  \frac{v}{V} 
\label{Eqn:Evbound}
\end{align}
The first inequality follows by Lemma~\ref{coupling} and because our choice of $N$ guarantees 
\begin{align}\label{eqn:base case}
\Prob(\mathcal{G}_{1}) = \Prob(\mathcal{C}^{1:p}_{1}(i)) > 1 -\frac{\delta \mu^{\star}}{2V}    
\end{align}
It is straightforward to check this using Hoeffding's inequality since $X^{1:N}_{0} \iidsim \mu_{0}$.

However, as noted above the application of Lemma \ref{coupling} is invalid unless $\Prob(\mathcal{G}^{c}_{v}) \leq \frac{\mu^{\star}}{4}$. We now show this holds for $v=1,\ldots,V$ by way of induction. Indeed, the base case holds by (\ref{eqn:base case}). Now suppose $\Prob(\mathcal{G}^{c}_{k}) \leq \frac{\mu^{\star}}{4}$ for $k = 1,\ldots,v-1$. Then we have
\begin{align*}
    \Prob(\mathcal{G}_{v}) &= \Prob(\mathcal{C}^{1:p}_{v}(i) | \mathcal{E}_{v-1}) \Prob(\mathcal{E}_{v-1}) \\ 
    &\geq  \Prob(\mathcal{E}_{v-1}) - \frac{\delta \mu^{\star}}{2V} \\
    & \geq 1 - \frac{\delta \mu^{\star}v}{V} \\
    &\geq 1 - \frac{\mu^{\star}}{4}
\end{align*}
The first inequality follows by Lemma~\ref{resampling_bound} and choice of $N$. The second inequality follows since the inductive hypothesis guarantees $\Prob(\mathcal{E}_{v-1}) \geq  1 - \delta  \mu^{\star}  \frac{v-1}{V}$ by (\ref{Eqn:Evbound}). It follows that $\Prob(\mathcal{G}^{c}_{v}) \leq \frac{\mu^{\star}}{4}$ for $v = 1,\ldots,V$. Hence (\ref{Eqn:Evbound}) also holds for all $v=1,\ldots,V$ and so 
     $\Prob(\mathcal{E}_{V}) \geq 1 - \delta  \mu^{\star}$
as required.
\end{proof}

\subsection{Proof of Theorem \ref{thm:maintheorem}}
The proof of Theorem~\ref{thm:maintheorem} uses the following additional lemma regarding the target random variables $X^{1:N}_{V}$.
\begin{lemma}\label{lem:CondExpectBound}
Let $\epsilon \in (0,\frac{1}{2})$ and recall $\lambda = \frac{\epsilon}{24V}$. For any $f$ such that $|f| \leq 1$,
\begin{align*}
    \left|\sum^{p}_{j=1} \E\left[f(\bar{X}^{i}_{V}) \mathbbm{1}_{A_{j}}(\bar{X}^{i}_{V}) \mid \mathcal{F}_{V-1} \right] - \frac{(1+\frac{\epsilon^{2}}{16})}{(1-\frac{\epsilon^{2}}{16})} \E_{\pi}[f] \right| \leq  \frac{3}{4}\epsilon
\end{align*}
\end{lemma}
\begin{proof}
Write $f(x) = f^{+}(x) - f^{-}(x)$, where $f^{+}(x) = \max\{f(x),0 \}$ and $f^{-}(x) = \max\{-f(x),0\}$. To be consistent with Proposition~\ref{resampling_error}, write $\lambda$ as $\frac{\tilde{\epsilon}}{6V}$ where $\tilde{\epsilon} = \epsilon / 4$.  For $j=1,\ldots,p$ we have
\begin{align*}
\E\left[f(\bar{X}^{i}_{V}) \mathbbm{1}_{A_{j}}(\bar{X}^{i}_{V}) \mid \mathcal{F}_{V-1} \right] 
&= \E_{\pi_{|A_{j}}}[f^{+}] \hat{P}^{j}_{V} - \E_{\pi_{|A_{j}}}[f^{-}] \hat{P}^{j}_{V} \\
&\leq \frac{(1+\tilde{\epsilon})}{(1-\tilde{\epsilon})} \E_{\pi}[f^{+}  \mathbbm{1}_{A_{j}}(\bar{X}^{i}_{V})] - \frac{(1-\tilde{\epsilon})}{(1+\tilde{\epsilon})}\E_{\pi}[f^{-} \mathbbm{1}_{A_{j}}(\bar{X}^{i}_{V})] \\
&\leq \frac{(1+\tilde{\epsilon}^{2})}{(1-\tilde{\epsilon}^{2})}\E_{\pi}[f \mathbbm{1}_{A_{j}}(\bar{X}^{i}_{V})] + \frac{2\tilde{\epsilon}}{(1-\tilde{\epsilon}^{2})} \pi(A_{j}) ,
\end{align*}
The first line follows by  Lemma~\ref{cond_exp_equality}, the second by application of Lemma~\ref{controlling_resampling} followed by Proposition~\ref{resampling_error} with our choice of $\lambda$, and the last since $|f| \leq 1$. Summing over $j$ gives
\begin{align*}
    \sum^{p}_{j=1} \E\left[f(\bar{X}^{i}_{V}) \mathbbm{1}_{A_{j}}(\bar{X}^{i}_{V}) \mid \mathcal{F}_{V-1} \right] \leq  \frac{(1+\frac{\epsilon^{2}}{16})}{(1-\frac{\epsilon^{2}}{16})} \E_{\pi}[f] + \frac{3}{4}\epsilon.
\end{align*}
The argument for the lower bound is identical.
\end{proof}

We are now ready to give the complete proof of Theorem~\ref{thm:maintheorem}.

\begin{proof}(Theorem \ref{thm:maintheorem})
Define the concentration event,
\begin{align*}
    \mathcal{D}^{j} = \left\{\left|\bar{f}^{j} - \E[ f(\bar{X}^{i}_{V}) \mathbbm{1}_{A_{j}}(\bar{X}^{i}_{V})| \mathcal{F}_{V-1}] \right| \leq \frac{\epsilon}{p}  \right\} ,
\end{align*}
where $\bar{f}^{j} = \frac{1}{N}\sum^{N}_{i=1}f(\bar{X}^{i}_{V})\mathbbm{1}_{A_{j}}(\bar{X}^{i}_{V})$. 
Note that over $\Ev$ we have
\begin{align*}
    \frac{1}{N}\sum^{N}_{i=1}f(X^{i}_{V}) =  \frac{1}{N}\sum^{N}_{i=1}\sum^{p}_{j=1}f(\bar{X}^{i}_{V})\mathbbm{1}_{A_{j}}(\bar{X}^{i}_{V}) = \sum^{p}_{j=1} \bar{f}^{j}
\end{align*}
%
By Lemma~\ref{lem:CondExpectBound}, it follows that over $\mathcal{E}_{V} \cap \mathcal{D}^{1} \cap \ldots \cap \mathcal{D}^{p}$,
\begin{align*}
    \frac{(1+\frac{\epsilon^{2}}{16})}{(1-\frac{\epsilon^{2}}{16})}\E_{\pi}[f] - \epsilon \leq \frac{1}{N}\sum^{N}_{i=1}f(X^{i}_{V})  \leq  \frac{(1+\frac{\epsilon^{2}}{16})}{(1-\frac{\epsilon^{2}}{16})}\E_{\pi}[f] + \epsilon
\end{align*}
It remains to show $\Prob(\mathcal{E}_{V} \cap \mathcal{D}^{1} \cap \ldots \cap \mathcal{D}^{p}) \geq \frac{3}{4}$. First, by Lemma \ref{conc_bound},
\begin{align*}
    \Prob(\cup^{p}_{j=1}(\mathcal{D}^{j})^{c}) \leq  \sum^{p}_{j=1}\Prob\left(\left|\bar{f}^{j} - \E[\bar{f}^{j}| \mathcal{F}_{V-1}]\right| > \frac{\epsilon}{p}  \right) 
    \leq \frac{1}{8},
\end{align*}
for $N > \frac{2p^{2}}{\epsilon^{2}}\log(32p)$. Note that $\Prob(\mathcal{E}_{V}) \geq  \frac{7}{8}$ by our choice of $N$ 
and $t$ after applying Lemma~\ref{OneStepError} with $\delta = \frac{1}{8}$. By the union bound,
\begin{align*}
    \Prob((\mathcal{D}^{1})^{c} \cup \ldots \cup (\mathcal{D}^{p})^{c} \cup \mathcal{E}^{c}_{V} ) \leq \frac{1}{8} + \frac{1}{8} = \frac{1}{4}.
\end{align*}
%
\end{proof}


\section{Applications}\label{applications}
In this section, we apply the results in the previous section to obtain bounds on the complexity of using sequential Monte Carlo to approximate expectations for two specific classes of multimodal distributions: mixtures of log-concave distributions and the mean-field Ising model. Each problem defines a nested sequence of distributions indexed by a dimension parameter $d$, and we consider the computational complexity, i.e. the growth of the computational cost with increasing $d$. We characterize the computational cost of applying the SMC algorithm given in Section~\ref{smc} by the total number of Markov transition steps $NVt$ required to apply Theorem~\ref{thm:maintheorem}, since the Markov transition steps are typically the most computationally expensive step of the algorithm. Throughout we will assume the corresponding partition sequence $\A(d)$ satisfies $\pi^{*} = \min_{j \in \{1,\ldots,p \}} \pi(A_{j})$ decays at most polynomially in $d$. (Otherwise, the contribution of $\E_{\pi}[f \mathbbm{1}_{A_{j^*}}]$ to $\E_{\pi}[f]$ for the corresponding $j^*$ is negligible for $d$ large since $\abs{f} \leq 1$.)

\subsection{Mixtures of Log-Concave Distributions}
\label{2_comp_Hs}
Consider the mixture distribution
\begin{align}\label{two_comp_halfspace}
\pi(x) \propto 
w \pi_{1}(x) \1_{H}(x) + (1-w)\pi_{2}(x) \1_{H^{c}}(x) \text{ for } x\in \mathbb{R}^{d},
\end{align}
where $w \in (0,1)$, $\pi_{j}(x) \propto e^{-f_{j}(x)}$, $f_{j}(x) > 0$ is a convex function for $j \in \{1,2\}$, and $H$ is a half-space. Throughout, we assume without loss of generality that $H$ and $H^{c}$ contain $\E_{\pi_{1}}[X]$ and $\E_{\pi_{2}}[X]$, respectively.
Note that $\pi$ is a sum of log-concave functions, since $H$ and $H^{c}$ being half-spaces means $\mathbbm{1}_{H}$ and $\mathbbm{1}_{H^{c}}$ are log-concave functions and log-concavity is closed under multiplication.

Suppose we are interested in sampling from $\pi(x)$ using SMC, and we set $\mu_{v}(x) \propto \pi^{\beta_{v}}(x)$ for $\beta_v \in (0,1]$ and $\beta_0,\ldots,\beta_V$ a strictly increasing sequence with $\beta_V = 1$. Since $f_{j}(x) > 0$, it follows that $w_{v}(x) \leq 1$ uniformly in $v$. We will assume that the following equality holds:
\begin{align}
\label{Applications: Equal Scales}
\int_{\mathbbm{R}^{d}}e^{-\beta_{v}f_{1}(x)} dx = \int_{\mathbbm{R}^{d}}e^{-\beta_{v}f_{2}(x)} dx  \text{ for } \beta_{v} \in (0,1]
\end{align}
For example, if $f_{j}(x) =2^{-1}\sigma_{j}^{-2}\norm{x- \mu_{j}}^{2}$, as below in Section~\ref{Application: Gaussian Mixture}, then (\ref{Applications: Equal Scales}) requires that $\sigma_{1} = \sigma_{2}$. More generally, (\ref{Applications: Equal Scales}) holds when $\pi_{1}$ and $\pi_{2}$ are log-concave product distributions that belong to the same location-scale family and have equal scales. Similar assumptions have been made in previous analyses of these types of problems (see \cite{woodard_1} and \cite{holden}), and it is unclear whether the assumption (\ref{Applications: Equal Scales}) can be dropped. For example, if $\sigma_{1} >\sigma_{2}$ in the problem considered in  Section~\ref{Application: Gaussian Mixture}, then $\mu^{\star}$ decays exponentially in $d$ \cite{woodard_2}.

Note that $\mu_{v}(x)$ can also be written as a sum of log-concave functions and assumption (\ref{Applications: Equal Scales}) gives
\begin{align*}
    \mu_{v}(H) &= \frac{w\int_{H} \pi_{1}^{\beta_{v}}(x) dx}{w\int_{\mathbb{R}^{d}} \pi_{1}^{\beta_{v}}(x) \mathbbm{1}_{H}(x) dx + (1-w)\int_{\mathbb{R}^{d}} \pi_{2}^{\beta_{v}}(x)\mathbbm{1}_{H^{c}}(x) dx} \\
    &= \frac{w\pi_{1,v}(H)}{w\pi_{1,v}(H) + (1-w)\pi_{2,v}(H^{c})}
\end{align*}
where $\pi_{v,j}(x) = e^{-\beta_{v} f_{j}(x)} / \int_{\mathbbm{R}^{d}}e^{-\beta_{v}f_{j}(x)} dx $
{
is an (untruncated) log-concave density function.} By \cite{lovasz_vempala_geom}, if $h(x)$ is a log-concave density and $H$ is a half-space containing the centroid of $h(x)$, then $ \int_{H} h(x) dx \geq e^{-1}$. Hence, $\mu^{\star} \geq \max\{w,1-w\} / e = \mathcal{O}(1)$.

Since $H$ and $H^{c}$ are convex sets and $\pi_{v,j}(x)$ is a log-concave density function, $\mu_{v|H}(x)$ and  $\mu_{v|H^{c}}(x)$  are log-concave density functions.
Upper bounds on the mixing time of various Markov kernels have been obtained in the literature for log-concave target distributions, assuming a warm start. For example, \cite{lovasz_vempala_geom} proved that both the ball walk and hit-and-run walk are rapidly mixing when the target distribution is a $C$-isotropic log-concave density function, and show that any $d$-dimensional log-concave density function can be brought into a $C$-isotropic position in $\mathcal{O}^{*}(d^{5})$ time (see \cite{vempala_geometric_random_walks} for a survey of these results). Recall that both $\mu^{\star} = \mathcal{O}(1)$ and $W = \mathcal{O}(1)$. It follows that by plugging into the bound on $N$ in Theorem~\ref{thm:maintheorem} we obtain
\begin{align}\label{polybound_N}
    \frac{N V}{\mu^{\star}} = \mathcal{O}\left(\frac{V^{3}Z^{2}}{\epsilon^{2}}  \cdot \log(V) \right) = \mathcal{O}(\text{poly}(V,\log(V),Z, \epsilon^{-1}))),
\end{align}
%
Therefore, by choosing $K_{v|A_{j}}$ to be either the ball walk or hit-and-run walk and applying the results of \cite{lovasz_vempala_geom}, we conclude 
\begin{align}\label{polybound_t}
     \max\limits_{j=1,\ldots,p} \max\limits_{v=1,\ldots,V} \tau_{v,j}\left( \frac{\mu^{\star}}{16  N  V}, 7\right)  = \mathcal{O}(\text{poly}(d, V,\log(V), Z, \epsilon^{-1}))
\end{align}
Combining  (\ref{polybound_N}) and (\ref{polybound_t}) gives that the computational cost of applying SMC to this problem is $\mathcal{O}(\text{poly}(d, V,\log(V), Z, \epsilon^{-1}))$.

\subsection{\textit{Gaussian mixtures}}
\label{Application: Gaussian Mixture}
We can provide more explicit bounds when $f_{j}(x)$ is known. Consider the Gaussian mixture distribution $f_{j}(x) = (2\sigma^{2})^{-1}\norm{x - (-1)^{j+1}\mathbf{1}_{d}}^{2}_{2}$ for $j \in \{1,2\}$ and $H = \{x: x^{\top} \mathbf{1}_{d} > 0 \}$ studied by \citet{woodard_1}. Here we assume $w = \frac{1}{2}$; similar results apply if $w \neq \frac{1}{2}$. Since the components share a common scale $\sigma>0$, it is straightforward to check that (\ref{Applications: Equal Scales}) holds and clearly both $\mathbf{1}_{d} \in H$ and $-\mathbf{1}_{d} \in H^{c}$. Define the sequence of inverse-temperatures 
\begin{align*}
\beta_{v} = \frac{1}{d}\left(1 + \frac{1}{d} \right)^{v} \text{ for } v = 0,\ldots, \lceil d \log(d) \rceil - 1 \text{ and } \beta_{V} = 1.
\end{align*}
Note $\beta_{0} = d^{-1}$, which ensures rapid mixing for initialization (see \cite{woodard_1} for details). 
To bound the normalizing constants, notice
\begin{align*}
    \frac{z_{v-1}}{z_{v}} &= \frac{\int_{\mathbbm{R}^{d}}  e^{-\beta_{v-1}f_{1}(x)} \mathbbm{1}_{H}(x) dx + \int_{\mathbbm{R}^{d}}  e^{-\beta_{v-1}f_{2}(x)} \mathbbm{1}_{H^{c}}(x) dx }{\int_{\mathbbm{R}^{d}}  e^{-\beta_{v}f_{1}(x)} \mathbbm{1}_{H}(x) dx + \int_{\mathbbm{R}^{d}}  e^{-\beta_{v}f_{2}(x)} \mathbbm{1}_{H^{c}}(x) dx} \\
    &= \frac{\Phi(\frac{\sqrt{d \beta_{v-1}}}{\sigma})\left(\frac{2\pi \sigma^{2}}{\beta_{v-1}}\right)^{\frac{d}{2}}}{\Phi(\frac{\sqrt{d \beta_{v}}}{\sigma})\left(\frac{2\pi \sigma^{2}}{\beta_{v}}\right)^{\frac{d}{2}}} \leq 2 \left(\frac{\beta_{v}}{\beta_{v-1}} \right)^{\frac{d}{2}}
\end{align*}
We see that $Z \leq 2e = \mathcal{O}(1)$. Next, we have
\begin{align*}
\mu_{v}(H) =   \frac{\int_{\mathbbm{R}^{d}} e^{-\beta_{v}f_{1}(x)} \mathbbm{1}_{ H}(x) dx}{ \int_{\mathbbm{R}^{d}}  e^{-\beta_{v}f_{1}(x)} \mathbbm{1}_{H}(x) dx + \int_{\mathbbm{R}^{d}}  e^{-\beta_{v}f_{2}(x)} \mathbbm{1}_{H^{c}}(x) dx} \geq \frac{1/2}{2} =
    \frac{1}{4}
\end{align*}
It follows that $\mu^{\star} \geq \frac{1}{4} = \mathcal{O}(1)$. Since $V = \mathcal{O}(d \log(d))$, combining with (\ref{polybound_N}) and (\ref{polybound_t}) implies that a computational complexity of
\begin{align*}
\mathcal{O}(\text{poly}(d,\log(d),\log(d \log(d)),\epsilon^{-1})) = \mathcal{O}^{*}(\text{poly}(d))
\end{align*}
for SMC on this problem, where $\mathcal{O}^{*}$ indicates that the dependence on logarithmic and error terms is suppressed.

\subsection{Mean Field Ising Model}\label{Application: mean field ising model}
Let $\mathcal{X} = \{-1,1 \}^{d}$ for $d \in \mathbbm{N}$. The mean field Ising model is defined by the distribution
\begin{align*}
    \pi(x) = \frac{1}{Z} \exp\left\{\frac{\alpha}{2d} \left( \sum^{d}_{i=1} x_{i}\right)^{2} \right\},
\end{align*}
with normalizing constant $Z = \sum_{x \in \mathcal{X}}\exp\{\frac{\alpha}{2d} ( \sum^{d}_{i=1} x_{i})^{2} \}$ and interaction parameter $\alpha \in \mathbbm{R}$. We will assume $d$ is odd as in \cite{woodard_1}.
Let $K_{v}$ be a Metropolis-Hastings chain with limiting distribution $\mu_{v}(x) \propto \pi^{\beta_{v}}(x)$. We assume $K_{v}$ uses a single site proposal distribution in which a coordinate $i \in \{1,\ldots,d\}$ is chosen uniformly and the sign of $x_{i}$ is flipped. Letting $V = d$ and  $\beta_{v} = v / d$, it follows that both the parallel tempering and simulated tempering algorithms are rapidly mixing \cite{madras_zheng,woodard_1}. We will show that SMC also provides a polynomial time approximation of $\E_{\pi}[f]$ for $|f| \leq 1$
using the same sequence of distributions, in that the bounds on $N, t$ required to apply Theorem \ref{thm:maintheorem} grow at most polynomially in $d$.
\par Our result is essentially a direct consequence of those from \citet{madras_zheng} and \citet{woodard_1}. Let $A_{1} = \{x : \sum^{d}_{i=1} x_{i} \geq 0 \}$ and
$A_{2} = \{x : \sum^{d}_{i=1} x_{i} < 0 \}$. First, note that sampling from $\mu_{0}(x) \propto 1$ is trivial. Now, notice that $\pi(x)$ is symmetric with respect to this partition for $\mu_{0},\ldots,\mu_{V}$ if $d$ is odd.
It follows that $\mu_{v}(A_{1}) = \mu_{v}(A_{2}) = 1/2$ for $v = 0,\ldots,V$ and so $\mu^{\star} = 1/2$.  \citet{woodard_1} show that
\begin{align*}
    \frac{\mu_{v}(x)}{\mu_{v-1}(x)} \in \left[\exp\left\{-\frac{\alpha}{2} \right\}, \exp\left\{\frac{\alpha}{2} \right\}\right]
\end{align*}
%
therefore $ZW \leq \exp\{\frac{\alpha}{2}\}$. Since $V = d$, $p=2$,  $\mu^{\star} = \mathcal{O}(1)$, and $ZW = \mathcal{O}(1)$, we have that  Theorem~\ref{thm:maintheorem} requires
\begin{align*}
N >  \frac{1}{\epsilon^{2}} \cdot \max\left\{c_{1} \cdot \left(\frac{VWZ}{\mu^{\star}} \right)^{2}  \log\left(\frac{c_{2} Vp}{ \mu^{\star}}\right), \ p^{2} \cdot \log\left( c_{3} p^{2}\right)  \right\} = \mathcal{O}^{*}(d^{2}),
\end{align*}
where $c_{1}$, $c_{2}$, and $c_{3}$ are positive constants. To obtain the required bound for $t$, note also that $NV / \mu^{\star} = \mathcal{O}(d^{3} \log(d))$.
A mixing time bound can be obtained \cite{madras, woodard_1} by defining the restricted kernel as in (\ref{Eqn: restricted kernel}) with corresponding invariant distributions $\mu_{v|A_{j}}$. \citet{madras_zheng} show that $\min_{v,j}\text{Gap}(K_{v|A_{j}})$ decays at most polynomially in $d$. Since for $j \in \{1,2\}$
we have (see \cite{marion})
\begin{align}
\label{warm_vs_gap}
\tau_{v,j}(\epsilon,M) &\leq \frac{1}{\text{Gap}(K_{v \mid A_{j}})}\left[\log(2\epsilon^{-1}) + \log(M - 1) \right] \\
&\leq \frac{1}{\min_{\substack{v^{\prime}=1,\ldots,V \\ j^{\prime}=1,\ldots,p}}\text{Gap}(K_{v^{\prime} \mid A_{j^{\prime}}})}\left[\log(2\epsilon^{-1}) + \log(M - 1) \right],
\end{align}
%
we obtain  $\tau_{v,j}(\epsilon,M) = \mathcal{O}(\text{poly}(d, \log(d)))$. Hence, $NVt = \mathcal{O}^{*}(\text{poly}(d))$. That is, the values of $N$ and $t$ required to apply Theorem \ref{thm:maintheorem} grow at most polynomially in $d$.

\section{Comparing SMC and Tempered MCMC}\label{smc_pt_compare}
Throughout this section, we assume that the sequence of interpolating distributions is a tempered sequence, i.e.,
\begin{align}\label{tempered_sequence}
    \mu_{v}(x) \propto \pi^{\beta_{v}}(x) \quad \text{ for }\; 0 < \beta_{0} < \ldots < \beta_{V} = 1,
\end{align}
and that $\mu_{v}(x) / \mu_{v-1}(x) \leq ZW$. We let $\pi^{\star} = \min_{j=1,\ldots,p}\pi(A_{j})$ and  $K_{v|A_{j}}$ a restricted kernel of the form  (\ref{Eqn: restricted kernel}).

Although population MC algorithms \cite{jasra_PMC} such as SMC have been argued to perform better on multimodal target distributions than existing MCMC algorithms, it is reasonable to wonder whether the sequential, parallel nature of SMC plays a role, or whether any such benefits are simply attributable to the use of the distribution sequence  (\ref{tempered_sequence}).  If the latter, then an MCMC algorithm that also takes advantage of the same distribution sequence may be expected to succeed (or fail) on the same problems where SMC succeeds (or fails).
The following results support this conjecture under fairly relaxed assumptions.

Ideally, we would like to show that the bounds on $N$ and $t$ in Theorem~\ref{thm:maintheorem} grow at most polynomially in the problem size if and only if the parallel and simulated tempering algorithms are rapidly mixing on the same problem. The results below allow us to partially resolve the `if' direction by relating the lower bounds of \cite{woodard_1} on the spectral gap of swapping chains to our results. To formalize this, we first relate certain quantities defined in \cite{woodard_1} to those appearing in Theorem \ref{thm:maintheorem}. Two key quantities appearing in the mixing time bounds of \citet{woodard_1} are the \textit{persistence} and the \textit{overlap}.  The following two propositions relate these quantities to the $\mu^*$ and $ZW$ quantities appearing in Theorem~\ref{thm:maintheorem}.
\begin{proposition}\label{mu_star_lb}
Define the \emph{persistence} \cite{woodard_2} of partition $\mathcal{A}$ under distribution sequence $\mu_0,\ldots,\mu_V$ by
\begin{align}
\label{persistence}
\gamma(\mathcal{A}) = \min_{j = 1,\ldots,p}\prod^{V}_{v=1}\min\left\{1,\frac{\mu_{v-1}(A_{j})}{\mu_{v}(A_{j})} \right\}
\end{align}
Then $\mu^{\star} > \gamma(\mathcal{A}) \cdot \pi^{\star}$
\end{proposition}
\begin{proof}
For any $v = 0,\ldots,V$ and $j = 1,\ldots,p$, we have that
\begin{align*}
    \mu_{v}(A_{j}) &= \frac{\mu_{v}(A_{j})}{\mu_{v+1}(A_{j})} \cdot \frac{\mu_{v+1}(A_{j})}{\mu_{v+2}(A_{j})}  \ldots \frac{\mu_{V-1}(A_{j})}{\mu_{V}(A_{j})} \cdot \pi(A_{j}) \geq \gamma(\mathcal{A}) \cdot \pi^{\star}
\end{align*}
Since this holds for any $v$ and $j$, it must hold for $\mu^{\star}$ as well.
\end{proof}

Controlling $\gamma(\mathcal{A})$ is critical to the success of the PT and ST algorithms (see \cite{woodard_2}). The quantity $\gamma(\mathcal{A})$ describes how the measure of the partition elements of $\mathcal{A}$ `persists' under tempering of the target distribution $\pi$. There is a close relationship between $\mu^{\star}$ and $\gamma(\mathcal{A})$. Recall that $\pi^{\star}$ decays at most polynomially in $d$ by assumption. It follows that polynomial decay in $\gamma(\mathcal{A})$ is enough to guarantee polynomial decay in $\mu^{\star}$ by Proposition~\ref{mu_star_lb}.

%
\begin{proposition}
\label{overlap_lb}
Define the \emph{overlap} \cite{woodard_1} for distribution sequence $\mu_0,\ldots,\mu_V$ and partition $\mathcal{A}$ by
%
\begin{align}\label{overlap}
    \delta(\mathcal{A}) = \min_{\substack{v = 0,\ldots,V-1 \\ j=1,\ldots,p}} \frac{\int_{A_{j}}\min\{\mu_{v}(x),\mu_{v+1}(x) \} dx}{\max\{\mu_{v}(A_{j}), \mu_{v+1}(A_{j}) \}}
\end{align}
Then,
\begin{align*}
    \delta(\mathcal{A}) \geq \min\left\{1,\frac{1}{ZW} \right\} \cdot \gamma(\mathcal{A}) \cdot \pi^{\star}
\end{align*}
\end{proposition}
\begin{proof}
Since $\mu_{v}(x) / \mu_{v+1}(x) \geq 1 / ZW$, it follows that for any $v = 0,\ldots,V-1$ and $j = 1,\ldots,p$,
\begin{align*}
    \int_{A_{j}} \min\left\{\mu_{v}(x), \mu_{v+1}(x) \right\} dx  &\geq \min\left\{1,\frac{1}{ZW} \right\} \cdot \mu_{v+1}(A_{j}) \geq \min\left\{1,\frac{1}{ZW} \right\} \cdot \mu^{\star}
\end{align*}
The result follows by Proposition \ref{mu_star_lb}.
\end{proof}
The overlap $\delta(\mathcal{A})$ controls the rate at which temperature jumps occur in the simulated tempering algorithm and lower bounds the probability of swapping states at adjacent temperatures in the parallel tempering algorithm (see \cite{woodard_1}). By Proposition \ref{overlap_lb}, if both $\gamma(\mathcal{A})$ and $(ZW)^{-1}$ decay at most polynomially in $d$, then the overlap decays at most polynomially.

Both $\delta(\mathcal{A})$ and $ZW$ can often be controlled in practice by decreasing the inverse temperature spacings.
Conversely, $\gamma(\mathcal{A})$ and $\mu^{\star}$ are highly dependent on the chosen \textit{form} of the interpolating densities
and are potentially problematic. For example, \citet{woodard_2} showed that $\gamma(\mathcal{A})$ decays exponentially for the truncated
Gaussian mixture model of Section~\ref{Application: Gaussian Mixture}
and for the mean field Pott's model, leading to torpid mixing of PT and ST in both cases. 

Applying Proposition~\ref{mu_star_lb} and (\ref{warm_vs_gap}) allows us to restate  Theorem~\ref{thm:maintheorem} with  $\gamma(\mathcal{A}) \cdot \pi^{\star}$ in place of $\mu^{\star}$ and the mixing times bounded by $\text{Gap}(K_{v \mid A_{j}})^{-1}$:
\begin{corollary}
\label{alternative_thm:maintheorem}
Suppose $|f| \leq 1$ and (\ref{assumption: density ratios}) holds. Let $\epsilon \in (0,\frac{1}{2})$ and choose
\begin{enumerate}
    \item  $ N > \frac{1}{\epsilon^{2}} \cdot  \max\left\{3456 \cdot  \left(\frac{VWZ}{\pi^{\star} \gamma(\mathcal{A})} \right)^{2} \cdot \log(\frac{64 Vp}{ \pi^{\star} \gamma(\mathcal{A})}), \ p^{2} \cdot  \log(1024  p^{2})   \right\}$
    \item $t > \log\left(\frac{288NV}{\gamma(\mathcal{A}) \pi^{\star}}\right) \cdot \left[\min_{\substack{v=1,\ldots,V \\ j=1,\ldots,p}}\textnormal{Gap}(K_{v|A_{j}})\right]^{-1}$
\end{enumerate}
Then with probability at least $3/4$,
\begin{align*}
     \left|\frac{1}{N}\sum^{N}_{i=1}f(X^{i}_{V})  - \frac{(1+\frac{\epsilon^{2}}{16})}{(1-\frac{\epsilon^{2}}{16})}\E_{\pi}[f]\right| \leq \epsilon
\end{align*}
\end{corollary}

Comparing the complexity results for swapping chain algorithms (see Theorem 3.1 of \cite{woodard_1}) to the bounds in Corollary \ref{alternative_thm:maintheorem}, we make the following observations. The dependence of $N$ on  $\gamma(\A)$ is quadratic (up to logarithmic terms) in Corollary \ref{alternative_thm:maintheorem} rather than exponential in $\abs{\A}$ as in \cite{woodard_1}. (This is similar to the results stated in \cite{holden}, who proved a lower bound on the spectral gap of a simulated tempering chain with quadratic dependence on the minimum measure of a given partition (similar to $\mu^{\star}$)). In addition, the dependence  in Corollary~\ref{alternative_thm:maintheorem} on the number of temperatures $V$ is quadratic rather than quartic, and dependence on $|\A|$ quadratic rather than cubic, compared to \cite{woodard_1} (omitting logarithmic terms).
(Note that $\text{Gap}(K_{0})$ appears in the bound of  \citet{woodard_1} but not in ours, but our assumption in Section~\ref{smc} that $\mu_0$ can be sampled efficiently is equivalent to assuming $\text{Gap}(K_{0})$ decays at most polynomially.) 

Unfortunately, we are currently unable to replace $\delta(\mathcal{A})$ in Corollary~\ref{alternative_thm:maintheorem} with  $ZW$. Proposition~\ref{overlap_lb} suggests that (at most) polynomial growth of $ZW$ may be a stronger requirement than polynomial decay of $\delta(\mathcal{A})$. 
This prevents us from concluding the `if' direction of our conjecture precisely. Nevertheless, we suspect that polynomial decay in $(ZW)^{-1}$ and $\delta(\mathcal{A})$ can be guaranteed on many problems of interest by decreasing the temperature spacings. We can summarize this discussion with the following corollary.


\begin{corollary}\label{smc_vs_pt}
Suppose $|f| \leq 1$, (\ref{assumption: density ratios}) holds, and $\textnormal{Gap}(K_{v|A_{j}})$, $\textnormal{Gap}(K_{0})$, $\gamma(\mathcal{A})$, $\pi^{\star}$, $(ZW)^{-1}$ decay at most polynomially in $d$ for $v = 1,\ldots,V = \mathcal{O}(\textnormal{poly}(d))$ and $j = 1,\ldots,p$. Then
\begin{enumerate}
    \item The parallel and simulated tempering algorithms provide a polynomial time $(\epsilon,\delta$)-randomized approximation for  $\E_{\pi}[f]$.
    \item The SMC algorithm given in Section \ref{smc} provides a polynomial time $(\epsilon,\delta$)-randomized approximation for $\E_{\pi}[f]$.
\end{enumerate}
\end{corollary}
\begin{proof}
This is a direct consequence of Corollary~\ref{alternative_thm:maintheorem} and the results from \cite{woodard_1}, along with the remarks regarding $\text{Gap}(K_{0})$ above. 
\end{proof}
Corollary \ref{smc_vs_pt} says that if the parallel and simulated tempering algorithms provide an efficient Monte Carlo approximation, then the SMC algorithm in Section~\ref{smc} will as well, up to the approximate equivalence of $(ZW)^{-1}$ and  $\delta(\mathcal{A})$ discussed above. (It is unclear at this time if the latter distinction represents a meaningful difference in the algorithms, or simply an artifact of different proof approaches.)  

Concluding the `only if' direction of our conjecture necessitates development of lower bounds on SMC which are comparable to those in \citet{woodard_2}.  This is of interest but beyond the scope of the current paper.

\section{Conclusion}\label{conclusion}
The upper bounds established in Theorem~\ref{thm:maintheorem} demonstrate that SMC algorithms can approximate global expectations using a polynomial number of Markov transition steps \textit{without} assuming globally mixing mutation kernels.  These results can be viewed as a generalization of those given in \cite{marion}: consider the special case of $\mathcal{A} = \mathcal{X}$ so that $\mu^{\star} = 1$ and $p = 1$, then the bounds stated in Theorem~\ref{thm:maintheorem} differ from those in \cite{marion} only by a $\mathcal{O}(V^{2})$ term appearing in the lower bound on $N$. To see why this occurs, note that when the mutation kernels globally mix, $X^{1:N}_{v}$ couples to an \textit{i.i.d.} sample from $\mu_{v}$ (see \cite{marion}). In this case, controlling $\hat{P}^{j}_{v}$ is unnecessary, since $N^{-1}\sum^{N}_{i=1} w_{v}(X^{i}_{v-1})$ always approximates $z_{v} / z_{v-1}$ for sufficiently large $N$ (conditional on the coupling event), regardless of whether  $\tilde{\mu}_{v-1}(A_{j})$ approximates $\mu_{v-1}(A_{j})$ well. That is, globally mixing mutation kernels `correct' for the bias of $N^{-1}\sum^{N}_{i=1} w_{v}(X^{i}_{v-1})$ even if the resampling at step $v-1$ goes terribly wrong. When global mixing cannot be assumed, it is necessary to control the trajectory of the resampling probability $\hat{P}^{j}_{v}$ across steps $v=1,\ldots,V$ in order to obtain  such a concentration result; the $\mathcal{O}(V^{2})$ term represents  the cost of  doing so.

In Section~\ref{applications} we applied our results to two well known example multimodal problems: sampling from mixtures of log-concave distributions and the mean field Ising model. We found that the total number of mutation steps $NVt$ required to apply Theorem~\ref{thm:maintheorem} is comparable to bounds on the spectral gap given in the literature for seemingly similar algorithms. We explored this relationship further in Section~\ref{smc_pt_compare} by comparing our results to those obtained in \cite{woodard_1} for the parallel and simulated tempering algorithms. Corollary~\ref{smc_vs_pt} almost suggests that if the complexity results of \cite{woodard_1} grow at a polynomial rate in the problem size, then so do the lower bounds on $N$ and $t$ given in Theorem~\ref{thm:maintheorem}. However, this conclusion cannot be made directly since we require uniform bounds on the density ratios. 
Establishing lower bounds on $Nt$ to obtain results comparable to \citet{woodard_2} is also of interest.


Finally, our results suggests that designing sequences of interpolating distributions for multimodal target distributions is an important problem. Indeed, Theorem~\ref{thm:maintheorem} describes all of the properties such sequences should possess. Namely, adjacent densities should be `close' to control $ZW$ while still ensuring $V$ polynomial in $d$, and modes should not be `lost' to control $\mu^{\star}$.
Naturally $\mu_{0}$ must also be easy to sample from. Unfortunately, finding such sequences is not straightforward. For example, annealed sequences break down on seemingly simple problems such as Gaussian mixtures with unequal scales (see \cite{woodard_2} and \cite{holden}). The effect of distribution sequence choice on the computational complexity of SMC is considered further elsewhere \cite{adaptiveSMC}.


\bibliography{Bibliography.bib}

\begin{thebibliography}{44}

\bibitem[\protect\citeauthoryear{Beskos et~al.}{2013}]{beskos}
\begin{barticle}[author]
\bauthor{\bsnm{Beskos},~\bfnm{A.}\binits{A.}},
  \bauthor{\bsnm{Jasra},~\bfnm{A.}\binits{A.}},
  \bauthor{\bsnm{Kantas},~\bfnm{N.}\binits{N.}} \AND
  \bauthor{\bsnm{Thiery},~\bfnm{A.}\binits{A.}}
(\byear{2013}).
\btitle{On the Convergence of Adaptive Sequential {M}onte {C}arlo Methods}.
\bjournal{Annals of Applied Probability}
\bvolume{26}.
\end{barticle}
\endbibitem

\bibitem[\protect\citeauthoryear{Capp\'{e} et~al.}{2012}]{Cappe}
\begin{barticle}[author]
\bauthor{\bsnm{Capp\'{e}},~\bfnm{O.}\binits{O.}},
  \bauthor{\bsnm{Guillin},~\bfnm{A.}\binits{A.}},
  \bauthor{\bsnm{Marin},~\bfnm{J.~M.}\binits{J.~M.}} \AND
  \bauthor{\bsnm{Robert},~\bfnm{C.~P.}\binits{C.~P.}}
(\byear{2012}).
\btitle{Population {M}onte {C}arlo}.
\bjournal{Journal of Computational and Graphical Statistics}
\bvolume{13}
\bpages{907--929}.
\end{barticle}
\endbibitem

\bibitem[\protect\citeauthoryear{Chopin}{2002}]{chopin_static}
\begin{barticle}[author]
\bauthor{\bsnm{Chopin},~\bfnm{N.}\binits{N.}}
(\byear{2002}).
\btitle{A sequential particle filter method for static models}.
\bjournal{Biometrika}
\bvolume{89}
\bpages{539--551}.
\end{barticle}
\endbibitem

\bibitem[\protect\citeauthoryear{Chopin}{2004}]{chopin_clt}
\begin{barticle}[author]
\bauthor{\bsnm{Chopin},~\bfnm{N.}\binits{N.}}
(\byear{2004}).
\btitle{Central limit theorem for sequential {M}onte {C}arlo methods and its
  application to {B}ayesian inference}.
\bjournal{Annals of Statistics}
\bvolume{32}
\bpages{2385--2411}.
\end{barticle}
\endbibitem

\bibitem[\protect\citeauthoryear{Del~Moral, Doucet and
  Jasra}{2006}]{delmoral_state}
\begin{barticle}[author]
\bauthor{\bsnm{Del~Moral},~\bfnm{P.}\binits{P.}},
  \bauthor{\bsnm{Doucet},~\bfnm{A.}\binits{A.}} \AND
  \bauthor{\bsnm{Jasra},~\bfnm{A.}\binits{A.}}
(\byear{2006}).
\btitle{Sequential {M}onte {C}arlo samplers}.
\bjournal{J. R. Stat. Soc}
\bvolume{68}
\bpages{411--436}.
\end{barticle}
\endbibitem

\bibitem[\protect\citeauthoryear{Douc and Capp\'{e}}{2005}]{douc}
\begin{barticle}[author]
\bauthor{\bsnm{Douc},~\bfnm{R.}\binits{R.}} \AND
  \bauthor{\bsnm{Capp\'{e}},~\bfnm{O.}\binits{O.}}
(\byear{2005}).
\btitle{Comparison of resampling schemes for particle filtering}.
\bjournal{ISPA 2005. Proceedings of the 4th International Symposium on Image
  and Signal Processing and Analysis, 2005}
\bpages{64--69}.
\end{barticle}
\endbibitem

\bibitem[\protect\citeauthoryear{Douc and Moulines}{2008}]{douc_lln}
\begin{barticle}[author]
\bauthor{\bsnm{Douc},~\bfnm{R.}\binits{R.}} \AND
  \bauthor{\bsnm{Moulines},~\bfnm{E.}\binits{E.}}
(\byear{2008}).
\btitle{Limit Theorems for Weighted Samples with Applications to Sequential
  {M}onte {C}arlo Methods}.
\bjournal{Annals of Statistics}
\bvolume{36}
\bpages{2344--2376}.
\end{barticle}
\endbibitem

\bibitem[\protect\citeauthoryear{Durham and Geweke}{2014}]{geweke}
\begin{barticle}[author]
\bauthor{\bsnm{Durham},~\bfnm{G.}\binits{G.}} \AND
  \bauthor{\bsnm{Geweke},~\bfnm{J.}\binits{J.}}
(\byear{2014}).
\btitle{Adaptive Sequential Posterior Simulators for Massively Parallel
  Computing Environments}.
\bjournal{{B}ayesian Model Comparison (Advances in Econometrics Vol. 34)}
\bpages{1--44}.
\end{barticle}
\endbibitem

\bibitem[\protect\citeauthoryear{Dwivedi et~al.}{2018}]{yuansi}
\begin{barticle}[author]
\bauthor{\bsnm{Dwivedi},~\bfnm{R.}\binits{R.}},
  \bauthor{\bsnm{Chen},~\bfnm{Y.}\binits{Y.}},
  \bauthor{\bsnm{Wainwright},~\bfnm{M.}\binits{M.}} \AND
  \bauthor{\bsnm{Yu},~\bfnm{B.}\binits{B.}}
(\byear{2018}).
\btitle{Log-concave sampling: {M}etropolis-Hastings algorithms are fast!}
\bjournal{Proceedings of the 31st Conference On Learning Theory}
\bvolume{75}.
\end{barticle}
\endbibitem

\bibitem[\protect\citeauthoryear{Eberle and Marinelli}{2010}]{eberle_1}
\begin{barticle}[author]
\bauthor{\bsnm{Eberle},~\bfnm{A.}\binits{A.}} \AND
  \bauthor{\bsnm{Marinelli},~\bfnm{C.}\binits{C.}}
(\byear{2010}).
\btitle{$L^{p}$ estimates for {F}eynman-{K}ac propagators with time-dependent
  reference measures}.
\bjournal{Journal of Mathematical Analysis and Applications}
\bvolume{365}
\bpages{120--134}.
\end{barticle}
\endbibitem

\bibitem[\protect\citeauthoryear{Eberle and Marinelli}{2013}]{eberle_2}
\begin{barticle}[author]
\bauthor{\bsnm{Eberle},~\bfnm{A.}\binits{A.}} \AND
  \bauthor{\bsnm{Marinelli},~\bfnm{C.}\binits{C.}}
(\byear{2013}).
\btitle{Quantitative approximations of evolving probability measures and
  sequential {M}arkov chain {M}onte {C}arlo methods}.
\bjournal{Probability Theory and Related Fields}
\bvolume{155}
\bpages{665--701}.
\end{barticle}
\endbibitem

\bibitem[\protect\citeauthoryear{Fearnhead and Taylor}{2010}]{fearnhead}
\begin{barticle}[author]
\bauthor{\bsnm{Fearnhead},~\bfnm{P.}\binits{P.}} \AND
  \bauthor{\bsnm{Taylor},~\bfnm{B.}\binits{B.}}
(\byear{2010}).
\btitle{An Adaptive Sequential {M}onte {C}arlo Sampler}.
\bjournal{{B}ayesian Analysis}
\bvolume{8}
\bpages{411--438}.
\end{barticle}
\endbibitem

\bibitem[\protect\citeauthoryear{Gelman et~al.}{2013}]{BDA}
\begin{barticle}[author]
\bauthor{\bsnm{Gelman},~\bfnm{A.}\binits{A.}},
  \bauthor{\bsnm{Carlin},~\bfnm{J.~B.}\binits{J.~B.}},
  \bauthor{\bsnm{Stern},~\bfnm{H.~S.}\binits{H.~S.}},
  \bauthor{\bsnm{Dunson},~\bfnm{D.~B.}\binits{D.~B.}},
  \bauthor{\bsnm{Vehtari},~\bfnm{A.}\binits{A.}} \AND
  \bauthor{\bsnm{Rubin},~\bfnm{D.~B.}\binits{D.~B.}}
(\byear{2013}).
\btitle{{B}ayesian Data Analysis, 3rd edition}.
\bjournal{Chapman \& Hall/CRC}.
\end{barticle}
\endbibitem

\bibitem[\protect\citeauthoryear{Geyer}{1991}]{geyer}
\begin{barticle}[author]
\bauthor{\bsnm{Geyer},~\bfnm{C.~J.}\binits{C.~J.}}
(\byear{1991}).
\btitle{{M}arkov chain {M}onte {C}arlo maximum likelihood}.
\bjournal{Computing Science and Statistics: Proceedings of the 23rd Symposium
  on the Interface}
\bvolume{23}
\bpages{156--163}.
\end{barticle}
\endbibitem

\bibitem[\protect\citeauthoryear{Jasra, Stephens and Holmes}{2007}]{jasra_PMC}
\begin{barticle}[author]
\bauthor{\bsnm{Jasra},~\bfnm{A.}\binits{A.}},
  \bauthor{\bsnm{Stephens},~\bfnm{D.}\binits{D.}} \AND
  \bauthor{\bsnm{Holmes},~\bfnm{C.}\binits{C.}}
(\byear{2007}).
\btitle{On population-based simulation for static inference}.
\bjournal{Statistics and Computing}
\bvolume{17}
\bpages{263--279}.
\end{barticle}
\endbibitem

\bibitem[\protect\citeauthoryear{Lee, Risteski and Ge}{2018}]{holden}
\begin{barticle}[author]
\bauthor{\bsnm{Lee},~\bfnm{H.}\binits{H.}},
  \bauthor{\bsnm{Risteski},~\bfnm{A.}\binits{A.}} \AND
  \bauthor{\bsnm{Ge},~\bfnm{R.}\binits{R.}}
(\byear{2018}).
\btitle{Beyond Log-concavity: Provable Guarantees for Sampling Multi-modal
  Distributions using Simulated Tempering {L}angevin {M}onte {C}arlo}.
\bjournal{Adv. in NeurIPS 31}
\bvolume{31}.
\end{barticle}
\endbibitem

\bibitem[\protect\citeauthoryear{Levin, Peres and Wilmer}{2009}]{Levin_Peres}
\begin{barticle}[author]
\bauthor{\bsnm{Levin},~\bfnm{D.~A.}\binits{D.~A.}},
  \bauthor{\bsnm{Peres},~\bfnm{Y.}\binits{Y.}} \AND
  \bauthor{\bsnm{Wilmer},~\bfnm{E.~L.}\binits{E.~L.}}
(\byear{2009}).
\btitle{{M}arkov chains and Mixing Times}.
\bjournal{American Mathematical Society}.
\end{barticle}
\endbibitem

\bibitem[\protect\citeauthoryear{Lov\'{a}sz and
  Vempala}{2006}]{lovasz_vempala_LC}
\begin{barticle}[author]
\bauthor{\bsnm{Lov\'{a}sz},~\bfnm{L.}\binits{L.}} \AND
  \bauthor{\bsnm{Vempala},~\bfnm{S.}\binits{S.}}
(\byear{2006}).
\btitle{Fast algorithms for logconcave functions: Sampling rounding,
  integration, and optimization}.
\bjournal{Proceedings of 47th Annual IEEE Symp. on Foundations of Computer
  Science}
\bpages{57--68}.
\end{barticle}
\endbibitem

\bibitem[\protect\citeauthoryear{Lov\'{a}sz and
  Vempala}{2007}]{lovasz_vempala_geom}
\begin{barticle}[author]
\bauthor{\bsnm{Lov\'{a}sz},~\bfnm{L.}\binits{L.}} \AND
  \bauthor{\bsnm{Vempala},~\bfnm{S.}\binits{S.}}
(\byear{2007}).
\btitle{The geometry of logconcave functions and sampling algorithms}.
\bjournal{Random Structures \& Algorithms}
\bvolume{30}
\bpages{307--358}.
\end{barticle}
\endbibitem

\bibitem[\protect\citeauthoryear{Madras and Randall}{2002}]{madras}
\begin{barticle}[author]
\bauthor{\bsnm{Madras},~\bfnm{N.}\binits{N.}} \AND
  \bauthor{\bsnm{Randall},~\bfnm{D.}\binits{D.}}
(\byear{2002}).
\btitle{{M}arkov Chain Decomposition for Convergence Rate Analysis}.
\bjournal{Annals of Applied Probability}
\bvolume{12}
\bpages{581--606}.
\end{barticle}
\endbibitem

\bibitem[\protect\citeauthoryear{Madras and Zheng}{2003}]{madras_zheng}
\begin{barticle}[author]
\bauthor{\bsnm{Madras},~\bfnm{N.}\binits{N.}} \AND
  \bauthor{\bsnm{Zheng},~\bfnm{Z.}\binits{Z.}}
(\byear{2003}).
\btitle{On the swapping algorithm}.
\bjournal{Random Structures \& Algorithms}
\bvolume{22}
\bpages{67--97}.
\end{barticle}
\endbibitem

\bibitem[\protect\citeauthoryear{Mangoubi, Pillai and Smith}{2021}]{mangoubi}
\begin{barticle}[author]
\bauthor{\bsnm{Mangoubi},~\bfnm{O.}\binits{O.}},
  \bauthor{\bsnm{Pillai},~\bfnm{N.}\binits{N.}} \AND
  \bauthor{\bsnm{Smith},~\bfnm{A.}\binits{A.}}
(\byear{2021}).
\btitle{Simple conditions for metastability of continuous {M}arkov chains}.
\bjournal{Journal of Applied Probability}
\bvolume{58}
\bpages{83--105}.
\end{barticle}
\endbibitem

\bibitem[\protect\citeauthoryear{Marinari and Parisi}{1992}]{Marinari}
\begin{barticle}[author]
\bauthor{\bsnm{Marinari},~\bfnm{E.}\binits{E.}} \AND
  \bauthor{\bsnm{Parisi},~\bfnm{G.}\binits{G.}}
(\byear{1992}).
\btitle{Simulated tempering: A new {M}onte {C}arlo scheme}.
\bjournal{Europhys. Lett. EPL}
\bvolume{19}
\bpages{451--458}.
\end{barticle}
\endbibitem

\bibitem[\protect\citeauthoryear{Marion, Mathews and Schmidler}{2021}]{marion}
\begin{barticle}[author]
\bauthor{\bsnm{Marion},~\bfnm{J.}\binits{J.}},
  \bauthor{\bsnm{Mathews},~\bfnm{J.}\binits{J.}} \AND
  \bauthor{\bsnm{Schmidler},~\bfnm{S.}\binits{S.}}
(\byear{2021}).
\btitle{Finite Sample Complexities for Sequential {M}onte {C}arlo Estimators}.
\bjournal{arXiv:1803.09365 (Submitted for Publication)}.
\end{barticle}
\endbibitem

\bibitem[\protect\citeauthoryear{Marion, Mathews and
  Schmidler}{2022}]{adaptiveSMC}
\begin{bunpublished}[author]
\bauthor{\bsnm{Marion},~\bfnm{J.}\binits{J.}},
  \bauthor{\bsnm{Mathews},~\bfnm{J.}\binits{J.}} \AND
  \bauthor{\bsnm{Schmidler},~\bfnm{S.}\binits{S.}}
(\byear{2022}).
\btitle{Finite Sample ${L}_{2}$ Bounds for Sequential {M}onte {C}arlo and
  Adaptive Path Selection}.
\bnote{Under revision}.
\end{bunpublished}
\endbibitem

\bibitem[\protect\citeauthoryear{Metropolis et~al.}{1953}]{metropolis}
\begin{barticle}[author]
\bauthor{\bsnm{Metropolis},~\bfnm{N.}\binits{N.}},
  \bauthor{\bsnm{Rosenbluth},~\bfnm{A.}\binits{A.}},
  \bauthor{\bsnm{Rosenbluth},~\bfnm{M.}\binits{M.}},
  \bauthor{\bsnm{Teller},~\bfnm{A.}\binits{A.}} \AND
  \bauthor{\bsnm{Teller},~\bfnm{E.}\binits{E.}}
(\byear{1953}).
\btitle{Equations of state calculations by fast computing machines}.
\bjournal{J. Chem. Phys.}
\bvolume{21}
\bpages{1087--1092}.
\end{barticle}
\endbibitem

\bibitem[\protect\citeauthoryear{Neal}{2001}]{annealed}
\begin{barticle}[author]
\bauthor{\bsnm{Neal},~\bfnm{R.}\binits{R.}}
(\byear{2001}).
\btitle{Annealed importance sampling}.
\bjournal{Statistics and Computing}
\bvolume{11}
\bpages{125--139}.
\end{barticle}
\endbibitem

\bibitem[\protect\citeauthoryear{Neal}{2011}]{neal_HMC}
\begin{barticle}[author]
\bauthor{\bsnm{Neal},~\bfnm{R.}\binits{R.}}
(\byear{2011}).
\btitle{{MCMC} Using {H}amiltonian Dynamics}.
\bjournal{Handbook of {M}arkov Chain {M}onte {C}arlo, Ch. 5}.
\end{barticle}
\endbibitem

\bibitem[\protect\citeauthoryear{P., A and A.}{2012}]{adaptive_clt}
\begin{barticle}[author]
\bauthor{\bsnm{P.},~\bfnm{Del~Moral}\binits{D.~M.}},
  \bauthor{\bsnm{A},~\bfnm{Doucet}\binits{D.}} \AND
  \bauthor{\bsnm{A.},~\bfnm{Jasra}\binits{J.}}
(\byear{2012}).
\btitle{On adaptive resampling strategies for sequential {M}onte {C}arlo
  methods}.
\bjournal{Bernoulli}
\bvolume{18}
\bpages{252--278}.
\end{barticle}
\endbibitem

\bibitem[\protect\citeauthoryear{Paulin, Jasra and Thiery}{2019}]{jasra}
\begin{barticle}[author]
\bauthor{\bsnm{Paulin},~\bfnm{D.}\binits{D.}},
  \bauthor{\bsnm{Jasra},~\bfnm{A.}\binits{A.}} \AND
  \bauthor{\bsnm{Thiery},~\bfnm{A.}\binits{A.}}
(\byear{2019}).
\btitle{Error bounds for sequential {M}onte {C}arlo samplers for multimodal
  distributions}.
\bjournal{Bernoulli}
\bvolume{25}
\bpages{310--340}.
\end{barticle}
\endbibitem

\bibitem[\protect\citeauthoryear{Rudoy and Wolfe}{2016}]{rudoy}
\begin{barticle}[author]
\bauthor{\bsnm{Rudoy},~\bfnm{D.}\binits{D.}} \AND
  \bauthor{\bsnm{Wolfe},~\bfnm{P.~J.}\binits{P.~J.}}
(\byear{2016}).
\btitle{{M}onte {C}arlo Methods for Multi-Modal Distributions}.
\bjournal{2006 Fortieth Asilomar Conference on Signals, Systems and Computers}
\bpages{2019--2023}.
\end{barticle}
\endbibitem

\bibitem[\protect\citeauthoryear{Salomone et~al.}{2018}]{nested_smc}
\begin{barticle}[author]
\bauthor{\bsnm{Salomone},~\bfnm{R.}\binits{R.}},
  \bauthor{\bsnm{South},~\bfnm{L.}\binits{L.}},
  \bauthor{\bsnm{Drovandi},~\bfnm{C.~C.}\binits{C.~C.}} \AND
  \bauthor{\bsnm{Kroese},~\bfnm{D.~P.}\binits{D.~P.}}
(\byear{2018}).
\btitle{Unbiased and Consistent Nested Sampling via Sequential {M}onte
  {C}arlo}.
\bjournal{arXiv: 1805.03924}.
\end{barticle}
\endbibitem

\bibitem[\protect\citeauthoryear{Schweizer}{2011}]{schweizer}
\begin{barticle}[author]
\bauthor{\bsnm{Schweizer},~\bfnm{N.}\binits{N.}}
(\byear{2011}).
\btitle{Non-asymptotic error bounds for sequential {MCMC} methods}.
\bjournal{PhD thesis, University of Bonn}.
\end{barticle}
\endbibitem

\bibitem[\protect\citeauthoryear{Schweizer}{2012}]{schweizer_mm}
\begin{barticle}[author]
\bauthor{\bsnm{Schweizer},~\bfnm{N.}\binits{N.}}
(\byear{2012}).
\btitle{Non-asymptotic error bounds for sequential {MCMC} methods in multimodal
  settings}.
\bjournal{arXiv:1205.6733}.
\end{barticle}
\endbibitem

\bibitem[\protect\citeauthoryear{Tan}{2015}]{tan}
\begin{barticle}[author]
\bauthor{\bsnm{Tan},~\bfnm{Z.}\binits{Z.}}
(\byear{2015}).
\btitle{Resampling {M}arkov Chain {M}onte {C}arlo Algorithms: Basic Analysis
  and Empirical Comparisons}.
\bjournal{Journal of Computational and Graphical Statistics}
\bvolume{24}
\bpages{328--356}.
\end{barticle}
\endbibitem

\bibitem[\protect\citeauthoryear{Vempala}{2005}]{vempala_geometric_random_walks}
\begin{barticle}[author]
\bauthor{\bsnm{Vempala},~\bfnm{S.}\binits{S.}}
(\byear{2005}).
\btitle{Geometric random walks: A survey}.
\bjournal{Combinatorial and Computational Geometry}.
\end{barticle}
\endbibitem

\bibitem[\protect\citeauthoryear{Wan and Zabaras}{2011}]{wan_and_zabaras}
\begin{barticle}[author]
\bauthor{\bsnm{Wan},~\bfnm{J.}\binits{J.}} \AND
  \bauthor{\bsnm{Zabaras},~\bfnm{N.}\binits{N.}}
(\byear{2011}).
\btitle{A {B}ayesian approach to multiscale inverse problems using the
  sequential {M}onte {C}arlo method}.
\bjournal{Inverse Problems}
\bvolume{27}.
\end{barticle}
\endbibitem

\bibitem[\protect\citeauthoryear{Wang, Machta and Katzgraber}{2015}]{PA_PT_SMC}
\begin{barticle}[author]
\bauthor{\bsnm{Wang},~\bfnm{W.}\binits{W.}},
  \bauthor{\bsnm{Machta},~\bfnm{J.}\binits{J.}} \AND
  \bauthor{\bsnm{Katzgraber},~\bfnm{H.}\binits{H.}}
(\byear{2015}).
\btitle{Comparing {M}onte {C}arlo methods for finding ground states of {I}sing
  spin glasses: Population annealing, simulated annealing, and parallel
  tempering}.
\bjournal{Phys. Rev. E}
\bvolume{92}.
\end{barticle}
\endbibitem

\bibitem[\protect\citeauthoryear{Weigel et~al.}{2017}]{PA_parallel}
\begin{barticle}[author]
\bauthor{\bsnm{Weigel},~\bfnm{M.}\binits{M.}},
  \bauthor{\bsnm{Barash},~\bfnm{L.~V.}\binits{L.~V.}},
  \bauthor{\bsnm{Borovsky},~\bfnm{M.}\binits{M.}},
  \bauthor{\bsnm{Janke},~\bfnm{W.}\binits{W.}} \AND
  \bauthor{\bsnm{Shchur},~\bfnm{L.~N.}\binits{L.~N.}}
(\byear{2017}).
\btitle{Population annealing: Massively parallel simulations in statistical
  physics}.
\bjournal{Journal of Physics: Conference Series}
\bvolume{921}.
\end{barticle}
\endbibitem

\bibitem[\protect\citeauthoryear{Weigel et~al.}{2021}]{PA_2}
\begin{barticle}[author]
\bauthor{\bsnm{Weigel},~\bfnm{M.}\binits{M.}},
  \bauthor{\bsnm{Barash},~\bfnm{L.}\binits{L.}}, \bauthor{\bsnm{Shchur}} \AND
  \bauthor{\bsnm{Janke},~\bfnm{W.}\binits{W.}}
(\byear{2021}).
\btitle{Understanding population annealing {M}onte {C}arlo simulations}.
\bjournal{Phys. Rev. E.}
\bvolume{103}.
\end{barticle}
\endbibitem

\bibitem[\protect\citeauthoryear{Whiteley}{2012}]{whiteley}
\begin{barticle}[author]
\bauthor{\bsnm{Whiteley},~\bfnm{N.}\binits{N.}}
(\byear{2012}).
\btitle{Sequential {M}onte {C}arlo samplers: error bounds and insensitivity to
  initial conditions}.
\bjournal{Stochastic Analysis and Applications}
\bvolume{30}
\bpages{774--798}.
\end{barticle}
\endbibitem

\bibitem[\protect\citeauthoryear{Woodard, Schmidler and
  Huber}{2009a}]{woodard_1}
\begin{barticle}[author]
\bauthor{\bsnm{Woodard},~\bfnm{D.}\binits{D.}},
  \bauthor{\bsnm{Schmidler},~\bfnm{S.}\binits{S.}} \AND
  \bauthor{\bsnm{Huber},~\bfnm{M.}\binits{M.}}
(\byear{2009}a).
\btitle{Conditions for rapid mixing of parallel and simulated tempering on
  multimodal distributions}.
\bjournal{Annals of Applied Probability}
\bvolume{19}
\bpages{617--640}.
\end{barticle}
\endbibitem

\bibitem[\protect\citeauthoryear{Woodard, Schmidler and
  Huber}{2009b}]{woodard_2}
\begin{barticle}[author]
\bauthor{\bsnm{Woodard},~\bfnm{D.}\binits{D.}},
  \bauthor{\bsnm{Schmidler},~\bfnm{S.}\binits{S.}} \AND
  \bauthor{\bsnm{Huber},~\bfnm{M.}\binits{M.}}
(\byear{2009}b).
\btitle{Sufficient Conditions for Torpid Mixing of Parallel and Simulated
  Tempering}.
\bjournal{Electronic Journal of Probability}
\bvolume{14}
\bpages{780--804}.
\end{barticle}
\endbibitem

\bibitem[\protect\citeauthoryear{Wu, Schmidler and Y.}{2021}]{keru}
\begin{barticle}[author]
\bauthor{\bsnm{Wu},~\bfnm{K.}\binits{K.}},
  \bauthor{\bsnm{Schmidler},~\bfnm{S.}\binits{S.}} \AND
  \bauthor{\bsnm{Y.},~\bfnm{Chen}\binits{C.}}
(\byear{2021}).
\btitle{Minimax Mixing Time of the {M}etropolis-Adjusted {L}angevin Algorithm
  for Log-Concave Sampling}.
\bjournal{arXiv: 2109.13055}.
\end{barticle}
\endbibitem

\end{thebibliography}


\begin{thebibliography}{9}

\bibitem{beskos}
Beskos, A., Jasra, A., Kantas, N., and Thiery, A. On the Convergence of Adaptive Sequential Monte Carlo Methods. \textit{Annals of Applied Probability}. 26(2) (2013).

\bibitem{PA_GPU}
Barash, L.Y., Weigel, M., Borovsk\'{y}, M., Janke, W., and Shchur, L. GPU accelerated population annealing algorithm. \textit{Computer Physics Communications}. 220 341-350 (2017).

\bibitem{Cappe}
Capp\'{e}, O., Guillin, A., Marin, J.M., and Robert, C.P. Population Monte Carlo. \textit{Journal of Computational and Graphical Statistics}. 13(4) 907-929 (2012).

\bibitem{chopin_static}
Chopin, N. A sequential particle filter method for static models. \textit{Biometrika}. 89, 539–551 (2002).

\bibitem{chopin_clt}
Chopin, N. Central limit theorem for sequential Monte Carlo methods and its
application to Bayesian inference. \textit{Annals of Statistics}. (32) 2385–2411 (2004).

\bibitem{delmoral_state}
Del Moral, P. Doucet, A. Jasra, A. Sequential Monte Carlo samplers.
\textit{J. R. Stat. Soc}. B 68(3), 411–436 (2006).

\bibitem{adaptive_clt}
Del Moral P, Doucet A, Jasra A. On adaptive resampling strategies for sequential
Monte Carlo methods. \textit{Bernoulli}. 18(1) 252-278 (2012).

\bibitem{douc}
Douc, R. and Capp\'{e}, O. Comparison of resampling schemes for particle filtering. \textit{ISPA 2005. Proceedings of the 4th International Symposium on Image and Signal Processing and Analysis, 2005}. 64-69  (2005).

\bibitem{douc_lln}
Douc, R. and Moulines, E. Limit Theorems for Weighted Samples with Applications to Sequential Monte Carlo Methods. \textit{Annals of Statistics}. 36(5) 2344-2376 (2008).

\bibitem{doucet_book}
Doucet, A., de Freitas, N. Gordon, N. Sequential Monte Carlo Methods in Practice. \textit{Statistics for Engineering and Information Science}. Springer. (2001).

\bibitem{geweke}
Durham, G. and Geweke, J. Adaptive Sequential Posterior Simulators for Massively Parallel Computing Environments. \textit{Bayesian Model Comparison (Advances in Econometrics Vol. 34)}. Emerald Group Publishing Limited, Bingley. 1-44 (2014).

\bibitem{yuansi}
Dwivedi, R., Chen, Y., Wainwright, M., and Yu, B. Log-concave sampling: Metropolis-Hastings algorithms are fast! \textit{Proceedings of the 31st Conference On Learning Theory}. 75 (2018).

\bibitem{eberle_1}
Eberle, A. and Marinelli, C. $L^{p}$ estimates for Feynman-Kac propagators with time-dependent reference measures. \textit{Journal of Mathematical Analysis and Applications}. 365 120-134 (2010).

\bibitem{eberle_2}
Eberle, A. and Marinelli, C. Quantitative approximations of evolving probability measures and sequential Markov chain Monte Carlo methods. \textit{Probability Theory and Related Fields}. 155 665–701 (2013).

\bibitem{fearnhead}
Fearnhead, P. and Taylor, B. An Adaptive Sequential Monte Carlo Sampler. \textit{Bayesian analysis}. 8 411-438 (2010).

\bibitem{BDA}
Gelman, A. Carlin, J.B., Stern, H.S., Dunson, D.B., Vehtari, A., and Rubin, D.B. Bayesian Data Analysis, 3rd edition. \textit{Chapman \& Hall/CRC}.

\bibitem{geyer}
Geyer, C. J. Markov chain Monte Carlo maximum likelihood. \textit{Computing Science and Statistics: Proceedings of the 23rd Symposium on the Interface}. 23 156–163 (1991).

\bibitem{PA_1}
Hukushima, K. and Iba, Y. Population Annealing and Its Application to a Spin Glass. \textit{AIP Conference Proceedings}. 690(1) (2003).

\bibitem{jasra_PMC}
Jasra, A., Stephens, D., and Holmes, C. On population-based simulation for static inference. \textit{Statistics and Computing}. 17 263-279 (2007).

\bibitem{holden}
Lee, H., Risteski, A., and Ge, R. Beyond Log-concavity: Provable Guarantees for Sampling Mult\scs{i}-modal Distributions using Simulated Tempering Langevin Monte Carlo. \textit{Adv. in NeurIPS 31}. (2018).

\bibitem{Levin_Peres}
Levin, D.A., Peres, Y., and Wilmer, E. L. Markov chains and Mixing Times. \textit{American Mathematical Society}. (2009).

\bibitem{financial_econ}
Lopes, H. and Tsay, R. Particle Filters and Bayesian Inference in
Financial Econometrics. \textit{J. Forecast}. 30, 168–209 (2011).

\bibitem{lovasz}
Lov\'{a}sz, L. and Vempala, S. Simulated annealing in convex bodies and an
$O^{*}(n^{4})$ volume algorithm. \textit{Journal of Computer and System Sciences}. 72(2) 392- 417 (2006).

\bibitem{lovasz vempala LC}
Lov\'{a}sz, L. and Vempala, S. Fast algorithms for logconcave functions: Sampling rounding, integration, and optimization. \textit{Proceedings of 47th Annual IEEE Symp. on Foundations of Computer Science}. 57-68 (2006).

\bibitem{lovasz_vempala_geom}
Lov\'{a}sz, L. and Vempala, S. The geometry of logconcave functions and sampling algorithms. \textit{Random Structures \& Algorithms}. 30(3) 307-358 (2007).

\bibitem{madras}
Madras, N. and Randall, D. Markov Chain Decomposition for Convergence Rate Analysis. \textit{Annals of Applied Probability}. 12(2) 581-606 (2002).

\bibitem{madras_zheng}
Madras, N. and Zheng, Z. On the swapping algorithm. \textit{Random Structures \& Algorithms}. 22 67-97 (2003).

\bibitem{mangoubi}
Mangoubi, O., Pillai, N., and Smith, A. Simple conditions for metastability of continuous Markov chains. \textit{Journal of Applied Probability} 58(1) 83-105 (2021). 

\bibitem{Marinari}
Marinari, E. and Parisi, G. Simulated tempering: A new Monte Carlo
scheme. \textit{Europhys. Lett. EPL}. 19 451–458 (1992).

\bibitem{marion}
Marion, J.,  Mathews, J., and Schmidler, S. Finite Sample Complexities for Sequential Monte Carlo Samplers. \textit{Submitted for Publication}. (2021).

\bibitem{metropolis}
Metropolis, N., Rosenbluth, A., Rosenbluth, M., Teller, A., and Teller, E. Equations
of state calculations by fast computing machines. \textit{ J. Chem. Phys.} 21(6) 1087–1092 (1953).

\bibitem{annealed}
Neal, R. Annealed importance sampling. \textit{Statistics and Computing}.
11(2):125-139 (2001).

\bibitem{neal HMC}
Neal, R. MCMC Using Hamiltonian
Dynamics. \textit{Handbook of Markov Chain Monte Carlo, Ch. 5  CRC Press}. (2011).

\bibitem{jasra}
Paulin, D., Jasra, A., and Thiery, A. Error bounds for sequential Monte Carlo samplers for multimodal distributions. \textit{Bernoulli}. 25 (1) 310 - 340 (2019). 

\bibitem{tracking}
Ristic, B., Arulampalam, S. Gordon, N. Beyond the Kalman Filter: Particle Filters for Tracking Applications. \textit{Artech House}. (2004).

\bibitem{coupling}
Roberts, G. and Rosenthal, J. General state space Markov chains and
MCMC algorithms. \textit{Probability Surveys}. 1 20-71 (2004).

\bibitem{rudoy}
Rudoy, D. and Wolfe, P.J. Monte Carlo Methods for Multi-Modal Distributions. \textit{2006 Fortieth Asilomar Conference on Signals, Systems and Computers}. 2019-2023 (2016).

\bibitem{nested smc}
Salomone, R., South, L., Drovandi, C.C., and Kroese, D. P. Unbiased and Consistent Nested Sampling via Sequential Monte Carlo. \textit{arXiv: 1805.03924}. (2018).

\bibitem{schweizer}
Schweizer, N. Non-asymptotic error bounds for sequential MCMC methods.
\textit{PhD thesis, University of Bonn}. (2011).

\bibitem{schweizer_mm}
Schweizer, N. Non-asymptotic error bounds for sequential MCMC methods in multimodal settings. \textit{arXiv:1205.6733}. (2012).

\bibitem{tan}
Tan, Z. Resampling Markov Chain Monte Carlo Algorithms: Basic Analysis and Empirical Comparisons. \textit{ Journal of Computational and Graphical Statistics}. 24(2) 328-356 (2015).

\bibitem{vempala_geometric_random_walks}
Vempala, S. Geometric random walks: A survey.
\textit{Combinatorial and Computational Geometry}. 2005.

\bibitem{parallel_mcmc}
VanDerwerken, D. and Schmidler, S. Parallel Markov Chain Monte Carlo. \textit{ArXiv
e-prints}. (2013).

\bibitem{wan and zabaras}
Wan, J. and Zabaras, N. A Bayesian approach to multiscale inverse problems using the sequential Monte Carlo method. \textit{Inverse Problems}. 27(10) (2011).

\bibitem{PA_PT_SMC}
Wang, W., Machta, J., and Katzgraber, H. Comparing Monte Carlo methods for finding ground states of Ising spin glasses: Population annealing, simulated annealing, and parallel tempering. \textit{Phys. Rev. E}. 92(1) (2015).

\bibitem{PA_parallel}
Weigel, M., Barash, L.V., Borovsky, M., Janke, W., and Shchur, L.N. Population annealing: Massively parallel simulations in statistical physics. \textit{Journal of Physics: Conference Series}. 921 (2017).

\bibitem{PA_2}
Weigel, M., Barash, L., Shchur, and Janke, W. Understanding population annealing Monte Carlo simulations. \textit{Phys. Rev. E}. 103(5) (2021).

\bibitem{whiteley}
Whiteley, N. Sequential Monte Carlo samplers: error bounds and insensitivity to
initial conditions. \textit{Stochastic Analysis and Applications}. 30(5) 774-798, (2012).

\bibitem{woodard_1}
Woodard, D., Schmidler, S., and Huber, M. Conditions for rapid mixing of parallel and simulated tempering on multimodal distributions. \textit{Annals of Applied Probability}. 19(2) 617-640 (2009).

\bibitem{woodard_2}
Woodard, D., Schmidler, S., and Huber, M. Sufficient Conditions for Torpid Mixing of Parallel and Simulated Tempering. \textit{Electronic Journal of Probability}. 14 780-804 (2009).

\bibitem{keru}
Wu, K., Schmidler, S., and Chen Y. Minimax Mixing Time of the Metropolis-Adjusted Langevin Algorithm for Log-Concave Sampling. \textit{arXiv: 2109.13055}. (2021).

\bibitem{bayesfactors}
Zhou, Y., Johansen, A., Aston, J. Toward Automatic Model Comparison: An Adaptive Sequential Monte Carlo Approach. \textit{J. Comp. \& Graph. Statistics}. 25(3), 701-726 (2016).

\end{thebibliography}
\bibliographystyle{imsart-nameyear.bst}

\appendix
\section{Appendix}
\label{appendix_a}
In this section, we provide the construction of the coupled random variables $\bar{X}^{1:N}_{v}$. Let $j \in \{1,\ldots,p\} $ and define
\begin{align*}
    \frac{dK^{t}_{v \mid A_{j}}(x, \cdot)}{d\rho_{x}} := f^{j}_{x} 
    \qquad \qquad \frac{d\mu_{v \mid A_{j}}}{d\rho_{x}} := g_{x}^{j},
\end{align*}
where $\rho_{x}$ is some dominating measure. The subscript denotes the implicit dependence on $x \in \mathcal{X}$, and the superscript denotes the partition element $A_{j}$. Now, let $h^{j}_{x} = \min\{g_{x}^{j},f^{j}_{x}\}$ and set
 \begin{align*}
     a^{j}_{x} = \int_{\mathcal{X}}h^{j}_{x} d\rho_{x}, \ \  b^{j}_{x} = \int_{\mathcal{X}} (f^{j}_{x} - h^{j}_{x}) d\rho_{x}, \ \ c^{j}_{x} = \int_{\mathcal{X}} (g_{x}^{j} - h^{j}_{x}) d\rho_{x}
 \end{align*}
Each $\bar{X}^{i}_{v}$ is constructed independently according to the following `coupling map' $C: \mathcal{X} \rightarrow \mathcal{X} \times \mathcal{X}$. Let $Z^{i}_{v \mid j}$, $U_{v \mid j}^{i}$, and $V^{i}_{v  \mid j}$ denote random variables whose distributions have densities $h^{j}_{x}/a^{j}_{x}$, $(f^{j}_{x} - h^{j}_{x})/b_{x}$, and $(g_{x}^{j} - h^{j}_{x})/c^{j}_{x}$, respectively, for $j \in 1,\ldots,p$. The map $C$ transitions $\tilde{X}^{i}_{v}$ to $(X^{i}_{v}, \bar{X}^{i}_{v})$ as follows
\begin{enumerate}
    \item Independently draw $Z^{i}_{v \mid \xi^{i}_{v}}$, $U^{i}_{v \mid \xi^{i}_{v}}$, and $V^{i}_{v \mid \xi^{i}_{v}}$.
    \item Set $(X^{i}_{v}, \bar{X}^{i}_{v}) = (Z^{i}_{v|\xi^{i}_{v}},Z^{i}_{v\mid \xi^{i}_{v}})$ with probability $a^{\xi^{i}_{v}}_{\tilde{X}^{i}_{v}}$. Else, set $(X^{i}_{v}, \bar{X}^{i}_{v}) = (U^{i}_{v|\xi^{i}_{v}},V^{i}_{v\mid \xi^{i}_{v}})$. 
\end{enumerate}
Then for $x \in A_{j}$ we have
\begin{align*}
    \Prob(C(x) \in B \times \mathcal{X}) &=  \Prob(U^{i}_{v\mid j} \in B \, \cap \, V^{i}_{v\mid j} \in \mathcal{X})(1 - a^{j}_{x}) + \Prob(Z^{i}_{v\mid j} \in B \, \cap \,  \mathcal{X})a^{j}_{x}, \\
    &= \Prob(U^{i}_{v\mid j} \in B)\Prob(V^{i}_{v\mid j} \in \mathcal{X})(1 - a^{j}_{x}) + \Prob(Z^{i}_{v\mid j} \in B)a^{j}_{x} \\
    &= K^{t}_{v\mid A_{j}}(x,B)
\end{align*}
and similarly, for $x \in A_{j}$ we have that $\Prob(C(x) \in \mathcal{X} \times B) = \mu_{v\mid A_{j}}(B)$.

We see then have that the new state $(X^{i}_{v}, \bar{X}^{i}_{v})$ satisfies 
\begin{align*}
    \Prob(X^{i}_{v} \in B \mid \tilde{X}^{i}_{v} \in A_{j}) = \int_{\mathcal{X}} \tilde{\mu}_{v|A_{j}}(dx)K_{v|A_{j}}^{t}(x,B) = \hat{\mu}_{v\mid A_{j}}(B)
\end{align*}
and $\bar{\mu}_v(B \mid \tilde{X}_v^i \in A_j) = \Prob(\bar{X}^{i}_{v} \in B \mid  \tilde{X}^{i}_{v} \in A_{j}) = \mu_{v \mid 
A_{j}}(B)$. By the law of total probability, it follows that $\mathcal{L}(X^{i}_{v}) = \hat{\mu}_{v}$. That is, even though the constructions are performed conditionally on $\{\tilde{X}^{i}_{v} \in A_{j} \}$, the marginal law of $X^{i}_{v}$ is preserved. However, note that this does not apply to coupled variable $\bar{X}_v^i$; that is, marginally $\mathcal{L}(\bar{X}^{i}_{v}) \neq \mu_{v}$ since in general $\Prob(\tilde{X}^{i}_{v} \in A_{j}) \neq \mu_{v}(A_{j})$. Instead we have
\begin{equation}
\mathcal{L}(\bar{X}^{i}_{v}) = \sum_{j=1}^p  \mu_{v \mid A_j} \Prob(\tilde{X}^{i}_{v} \in A_{j}) \neq 
\sum_{j=1}^p  \mu_{v \mid A_j} \mu_v(A_j)= \mu_v.
\label{Eqn:IneqLawBarX}
\end{equation}
Note also that by construction, $\Prob(X^{i}_{v} = \bar{X}^{i}_{v} \mid  \tilde{X}^{i}_{v} \in A_{j}) = \|\hat{\mu}_{v \mid  A_{j}}(\cdot) - \mu_{v\mid A_{j}}(\cdot)\|_{\text{TV}}$.

The next lemma shows that the random variables produced by $C$ satisfy certain useful conditional independencies.
\begin{lemma}\label{cond_ind_statements}
Let $[-i] = \{1,\ldots,i-1,i+1,\ldots,N\}$. For $i = 1,\ldots,N$ we have:
\begin{enumerate}[label=(\alph*)]
    \item $\bar{X}^{i}_{v} \ind \tilde{X}^{[-i]}_{v} \mid \tilde{X}^{i}_{v}$
    \item $\bar{X}^{i}_{v} \ind \bar{X}^{[-i]}_{v} \mid  \tilde{X}^{i}_{v}$ \label{bars}
    \item $\bar{X}^{i}_{v} \ind (X^{1:N}_{k}, \tilde{X}^{1:N}_{k})^{v-1}_{k=1} \mid  \tilde{X}^{i}_{v}$
\end{enumerate}
\end{lemma}
\begin{proof}
Let $W^{i}_{v} =  (Z^{i}_{v \mid \xi^{i}_{v}},\, U^{i}_{v \mid \xi^{i}_{v}} ,\, V^{i}_{v \mid  \xi^{i}_{v}})$. All three statements follow immediately from the fact that
\begin{align*}
   W^{i}_{v} \, \ind \, W^{[-i]}_{v} \mid \tilde{X}^{i}_{v} 
\end{align*}
for $i = 1,\ldots,N$. That is, $Z^{i}_{v \mid \xi^{i}_{v}}$ , $U^{i}_{v \mid \xi^{i}_{v}}$, and  $V^{i}_{v \mid \xi^{i}_{v}}$ are, by construction, drawn independently of $((W^{1:N}_{k}, \tilde{X}^{1:N}_{k})^{v-1}_{k=1}, W^{[-i]}_{v}, \tilde{X}^{[-i]}_{v})$
conditional on $\tilde{X}^{i}_{v}$ for each $i=1,\ldots,N$.
\end{proof}
Since the constructions are performed conditional on the event $\tilde{X}^{i}_{v} \in A_{j}$, the $\bar{X}^{1:N}_{v}$ random variables are conditionally independent but not marginally independent; this is because the $\tilde{X}^{1:N}_{v}$ random variables are dependent. Nonetheless, we can still refine Lemma~\ref{cond_ind_statements}(c) slightly as follows:
\begin{lemma}
\label{cond_ind}
For $B \subset \mathcal{X}$,
\begin{align*}
\Prob(\bar{X}^{i}_{v} \in B \mid \mathcal{F}_{v-1},\tilde{X}^{i}_{v} \in A_{j}) = \Prob(\bar{X}^{i}_{v} \in B \mid \tilde{X}^{i}_{v} \in A_{j})
\end{align*}
\end{lemma}
\begin{proof}
\begin{align*}
 \Prob(\bar{X}^{i}_{v} \in B \mid \mathcal{F}_{v-1},\tilde{X}^{i}_{v} \in A_{j}) 
 &=  \E[\Prob(\bar{X}^{i}_{v} \in B \mid \tilde{X}^{i}_{v}) \mid \mathcal{F}_{v-1},\tilde{X}^{i}_{v} \in A_{j}]\\
 &=  \E[\mu_{v\mid A_{\xi^{i}_{v}}}(B) \mid \mathcal{F}_{v-1}, \tilde{X}^{i}_{v} \in A_{j}] \\
 &= \mu_{v \mid A_{j}}(B) \\
 &= \Prob(\bar{X}^{i}_{v} \in B \mid \tilde{X}^{i}_{v} \in A_{j})
\end{align*}
The first equality follows by iterated expectations and Lemma \ref{cond_ind_statements}(c), the second by the definition of $C$, and the third since the expectation is constant upon conditioning on $\tilde{X}^{i}_{v} \in A_{j}$. The final equality follows since $\{\tilde{X}^{i}_{v} \in A_{j} \} = \{\bar{X}^{i}_{v} \in A_{j} \}$ by the definition of $C$.
\end{proof}
The next result is a consequence of Lemma~\ref{cond_ind} 
%
\begin{lemma}
\label{cond_exp_equality}
Let $f: \mathcal{X} \rightarrow \mathbb{R}$ be a $\Prob$-measurable function. Then
\begin{align}
\label{Lem:condexpect}
    \E[f(\bar{X}^{i}_{v}) \1_{A_{j}}(\bar{X}^{i}_{v}) \mid \mathcal{F}_{v-1}] = \E_{\mu_{v|A_{j}}}[f] \cdot \hat{P}^{j}_{v}
\end{align}
In particular,
\begin{align}
\label{Lem:condexpect_weights}
    \E[\bar{w}^{j}_{v+1} | \mathcal{F}_{v-1}] = \frac{z_{v+1}}{z_{v}} \cdot \frac{\mu_{v+1}(A_{j})}{\mu_{v}(A_{j})} \cdot \hat{P}^{j}_{v}
\end{align}
where $\bar{w}^{j}_{v+1} = N^{-1} \sum^{N}_{i=1} w_{v+1}(\bar{X}^{i}_{v})\1_{A_j}(\bar{X}_{v}^i)$ as before.
\end{lemma}
\noindent Note that if we had $\bar{X}_{v-1}^i$ distributed exactly according to $\mu_{v-1}$ then we would have $\E_{\mu_{v-1}}[w_v(\bar{X}_{v-1}^i)\1_{A_j}(\bar{X}_{v-1}^i) ] = \frac{z_v}{z_{v-1}}\mu_v(A_j)$.  However, $\mathcal{L}(\bar{X}_{v-1}^i) \neq \mu_{v-1}$ as noted by (\ref{Eqn:IneqLawBarX}) above. Lemma~\ref{cond_exp_equality} says however that if the sample history $\{(\tilde{X}_s^{1:N},X_s^{1:N},\bar{X}_s^{1:N})\}_{s=1}^{v-2}$ satisfies conditions (\ref{initial}-\ref{coupling_event}) for steps $1,\ldots,v-2$,
the sample estimator $\bar{w}_v^j$ of this quantity is unbiased up to a multiplicative error $\hat{P}_{v-1}^j/\mu_{v-1}(A_j)$ introduced by the fact that the marginal probability of resampling into $A_j$ at step $v-1$ is $\hat{P}_{v-1}^j$ rather than exactly $\mu_{v-1}(A_j)$.
\begin{proof}
We show (\ref{Lem:condexpect}) for the case of $f(\cdot) = w_{v+1}(\cdot)$; the argument for the general case is identical. The result (\ref{Lem:condexpect_weights}) follows by linearity of expectation. 

Assume that $\hat{P}^{j}_{v} \neq 0$, otherwise the result holds trivially since this implies $\E[\bar{w}^{j}_{v+1} | \mathcal{F}_{v-1}] = 0 = 
\hat{P}^{j}_{v}$ (recall the events $\tilde{X}_v^i \in A_j$ and $\bar{X}_v^i \in A_j$ are equivalent).
Then we compute that
\begin{align*}
\MoveEqLeft[6] \E[w_{v+1}(\bar{X}^{i}_{v}) \1_{\bar{X}^{i}_{v} \in A_{j}} \mid \mathcal{F}_{v-1}] & \\
     = &\E[w_{v+1}(\bar{X}^{i}_{v}) \mid \mathcal{F}_{v-1}, \tilde{X}^{i}_{v} \in A_{j}] \Prob(\tilde{X}^{i}_{v} \in A_{j} | \mathcal{F}_{v-1}) \\
     = &\E[w_{v+1}(\bar{X}^{i}_{v}) \mid \tilde{X}^{i}_{v} \in A_{j}] \Prob(\tilde{X}^{i}_{v} \in A_{j} | \mathcal{F}_{v-1}) \\
     = &\frac{z_{v+1}}{z_{v}} \frac{\mu_{v+1}(A_{j})}{\mu_{v}(A_{j})}\left[\frac{\sum^{N}_{i=1}w_{v}(X^{i}_{v-1}) \1_{X^{i}_{v-1} \in A_{j}}}{\sum^{N}_{i=1}w_{v}(X^{i}_{v-1})} \right] \\
     = &\frac{z_{v+1}}{z_{v}} \cdot \frac{\mu_{v+1}(A_{j})}{\mu_{v}(A_{j})} \cdot \hat{P}^{j}_{v}
\end{align*}
The second equality follows by Lemma~\ref{cond_ind}. The third equality follows since $\Prob(\bar{X}^{i}_{v} \in B \mid \tilde{X}^{i}_{v-1} \in A_{j}) = \bar{\mu}_{v}(B \mid  \tilde{X}^{i}_{v-1} \in A_{j}) = \mu_{v |A_{j}}(B)$
by the definition of $C$.
\end{proof}
We conclude with the following concentration bound, which is a consequence of the results derived so far.
\begin{lemma}\label{conc_bound}
Let $f: \mathcal{X} \rightarrow (a,b)$ be a $\Prob$-measurable bounded function
and $\bar{f} = N^{-1}\sum^{N}_{i=1} f(\bar{X}^{i}_{v})$. Then for $\lambda > 0$,
\begin{equation*}
    \Prob\left(\abs{\bar{f} - \E[\bar{f} \mid \mathcal{F}_{v-1}]} > \lambda \right) \leq  4e^{-\frac{N \lambda^{2}}{2 (b-a)^{2}}}
\end{equation*}
\end{lemma}
\begin{proof}
Note that $\E[\bar{f} \mid \mathcal{F}_{v-1}] = \E[\E[\bar{f} \mid \tilde{X}^{1:N}_{v}] \mid \mathcal{F}_{v-1}]$ by Lemma \ref{cond_ind_statements}. Applying the triangle inequality and the union bound gives
\begin{align*}
\MoveEqLeft \Prob\left(\left|\bar{f} - \E[\bar{f} \mid \mathcal{F}_{v-1}] \right| > \lambda \right) & \\
    \leq  &\Prob\left(\left|\bar{f} - \E[\bar{f} \mid \tilde{X}^{1:N}_{v}] \right| > \frac{\lambda}{2} \right) + \Prob\left(\left|\E[\bar{f} \mid \tilde{X}^{1:N}_{v}] -  \E[\E[\bar{f} \mid \tilde{X}^{1:N}_{v}] \mid \mathcal{F}_{v-1}] \right| > \frac{\lambda}{2} \right),
\end{align*}
Notice that
\begin{align}
\E[\bar{f} \mid \tilde{X}^{1:N}_v] = \E\left[\frac{1}{N}\sum^N_{i=1} f(\bar{X}^i_v) \Bigm|  \tilde{X}^{1:N}_v \right] = \frac{1}{N}\sum^N_{i=1}\E[f(\bar{X}^i_v) \mid \tilde{X}^i_v] 
\label{cond_mean}
\end{align}
where the second equality follows by Lemma \ref{cond_ind_statements}. Since $\bar{X}^{1}_{v} \ind \ldots \ind \bar{X}^{N}_{v} \mid \tilde{X}^{1:N}_{v}$ by Lemma \ref{cond_ind_statements}, the first term can be controlled by the conditional form of Hoeffding's inequality:
\begin{align*}
\Prob\left(\left|\bar{f} - \E[\bar{f} \mid \tilde{X}^{1:N}_{v}] \right| > \frac{\lambda}{2} \right) 
&= \E\left[ \Prob\left(\left|\bar{f} - \E[\bar{f} \mid \tilde{X}^{1:N}_{v}] \right| > \frac{\lambda}{2} \Bigm| \tilde{X}^{1:N}_{v}\right)\right] \\ 
&= \E\left[ \Prob\left(\left|\sum^{N}_{i=1} \left(f(\bar{X}^{i}_{v}) - \E[f(\bar{X}^{i}_{v}) \mid  \tilde{X}^{1:N}_{v}]\right) \right| > \frac{N\lambda}{2} \Bigm| \tilde{X}^{1:N}_{v}\right)\right] \\
    &\leq 2e^{-\frac{N \lambda^{2}}{2 (b-a)^{2}}}
\end{align*}
Next, note that the random variable
$\E[f(\bar{X}^{i}_{v})\mid \tilde{X}^{i}_{v}]$ is also a bounded function of $\tilde{X}^{i}_{v}$. But we always have that $\tilde{X}^{1}_{v} \ind \ldots \ind \tilde{X}^{N}_{v} \mid \mathcal{F}_{v-1}$, so we can reapply an identical argument using the conditional version of Hoeffding's inequality:
\begin{align*}
\Prob\left(\left|\sum^{N}_{i=1}\left(\E[f(\bar{X}^{i}_{v}) \mid \tilde{X}^{i}_{v}] -  \E[\,\E[f(\bar{X}^{i}_{v}) \mid \tilde{X}^{i}_{v}] \mid \mathcal{F}_{v-1}] \right) \right| > \frac{N\lambda}{2}\right) \leq  2e^{-\frac{N \lambda^{2}}{2 (b-a)^{2}}}
\end{align*}
It follows that
\begin{align*}
    \Prob\left(\left|\bar{f} - \E[\bar{f} | \mathcal{F}_{v-1}] \right| > \lambda \right) \leq 4e^{-\frac{N \lambda^{2}}{2 (b-a)^{2}}}
\end{align*}
\end{proof}

\end{document}